%% file: main.tex
\documentclass[11pt]{article}

\usepackage{amsmath,amsfonts,amssymb,amsthm}
\usepackage{amsthm}
\usepackage{graphicx,color}
\usepackage{framed}
\usepackage{enumitem}
\usepackage{thmtools}
\usepackage{thm-restate}
\usepackage{xspace}
\usepackage{bm}
\usepackage{todonotes}
\usepackage{footnote}
\usepackage{mathtools}
\usepackage{mathrsfs}
\usepackage[margin=1in]{geometry}
\usepackage{todonotes}
\usepackage[most]{tcolorbox}
\usepackage{footnote}
\usepackage{rotating}
\usepackage[pdftex, plainpages = false, pdfpagelabels, 
bookmarks=true,
bookmarksopen = true,
bookmarksnumbered = true,
breaklinks = true,
linktocpage,
pagebackref,
colorlinks = true,  % was true
linkcolor = blue,
urlcolor  = blue,
citecolor = red,
anchorcolor = green,
hyperindex = true,
hyperfigures
]{hyperref} 
\usepackage{tabularx}
\usepackage{tikz} 
\usetikzlibrary{calc}
\usepackage{thm-autoref}
\usepackage[nameinlink]{cleveref}
\newcommand{\hh}{\ensuremath{\mathcal{H}}}

\newtheorem{proposition}{Proposition}
\newtheorem{theorem}{Theorem}
\newtheorem{claim}{Claim}

\newtheorem{lemma}{Lemma}
\newtheorem{observation}{Observation}
\theoremstyle{definition}
\newtheorem{definition}{Definition}

\newtheorem{corollary}{Corollary}

\crefname{invar}{invariant}{invariants}
\crefname{ineq}{inequality}{inequalities}
\crefname{constr}{constraint}{constraints}
\crefname{tbl}{table}{tables}
\crefname{lem}{lemma}{lemmata}
\crefname{lemma}{lemma}{lemmata}
\crefname{cond}{condition}{conditions}

\newcommand{\lr}[1]{\left( #1\right)}
\newcommand{\LR}[1]{\left\{ #1\right\}}

\newcommand{\Oh}{\mathcal{O}}
\newcommand{\polyn}{n^{\Oh(1)}}

\newcommand{\cA}{\mathcal{A}}
\newcommand{\cG}{\mathcal{G}}
\newcommand{\cH}{\mathcal{H}}
\newcommand{\cF}{\mathcal{F}}
\newcommand{\cS}{\mathcal{S}}
\newcommand{\cI}{\mathcal{I}}

\newcommand{\FII}{\textup{\textsf{FII}}\xspace}
\newcommand{\FPT}{\textup{\textsf{FPT}}\xspace}
\newcommand{\OPT}{\textup{\textsf{OPT}}}

\newcommand{\md}{\mathbf{mod}}
\newcommand{\mdh}{\md_{\cH}}
\newcommand{\edh}{\mathbf{ed}_{\cH}}
\newcommand{\twh}{\mathbf{tw}_{\cH}}
\newcommand{\tw}{\mathbf{tw}}
\newcommand{\td}{\mathbf{td}}

\newcommand{\hhmd}[1][\hh]{\md_{#1}}
\newcommand{\hhtwfull}[1][\hh]{{#1}-treewidth}

\newcommand{\pmd}{\ensuremath{\mathcal{P}_{\md}}\xspace}
\newcommand{\deltof}{\textup{\textsc{Vertex}} \allowbreak \textup{\textsc{Deletion}} \allowbreak \textup{\textsc{ to }} \allowbreak \text{$\cF$} \xspace}
\newcommand{\FPTAS}{\textup{\textsf{FPT-AS}}\xspace}
\newcommand{\EPTAS}{\textup{\textsf{EPTAS}}\xspace}
\newcommand{\FPTASes}{\textup{\textsf{FPT-AS}es}\xspace}

\newcommand{\vc}{\textup{\textsc{Vertex Cover}}\xspace}
\newcommand{\indset}{\textup{\textsc{Independent Set}}\xspace}
\newcommand{\fpacking}{\textup{\textsc{$\cF$-Minor Packing}}\xspace}
\newcommand{\spacking}{\textup{\textsc{$\cS$-Subgraph Packing}}\xspace}
\newcommand{\domset}{\textup{\textsc{Dominating Set}}\xspace}
\newcommand{\anndomset}{\textup{\textsc{Annotated Dominating Set}}\xspace}

\newcommand{\tilS}{\widetilde{S}}
\newcommand{\ads}{\textup{\textsc{Blue-White Dominating Set}}\xspace}
\newcommand{\til}[1]{\widetilde{#1}}
\newcommand{\tilG}{\til{G}}
\newcommand{\tilt}{\til{t}}
\newcommand{\tilb}{\til{B}}
\newcommand{\cP}{\mathcal{P}}
\newcommand{\xvc}{\textup{\textsc{Set Intersecting Vertex Cover}}\xspace}

\title{FPT Approximations for Packing and Covering Problems Parameterized by 
Elimination Distance and Even Less}

\author{
	Tanmay Inamdar\thanks{
		Department of Informatics, University of Bergen, Norway.}
	\and
	Lawqueen Kanesh\thanks{Indian Institute of Technology, Jodhpur}
	\and
	Madhumita Kundu\addtocounter{footnote}{-2}\footnotemark{}
	\and
	M.~S.~Ramanujan\addtocounter{footnote}{2}\thanks{University of Warwick, United Kingdom} 
	\and
	Saket Saurabh\addtocounter{footnote}{-4}\footnotemark{} \addtocounter{footnote}{4}\thanks{Institute of Mathematical Sciences, Chennai. 
		\\Tanmay Inamdar and Madhumita Kundu are supported by the European Research Council (ERC) under the European Union’s Horizon 2020 research and innovation programme (grant agreement No. 819416). M. S. Ramanujan is supported by Engineering and Physical Sciences Research Council (EPSRC) grants EP/V007793/1 and EP/V044621/1. Saket Saurabh is supported by the European Research Council (ERC) under the European Union’s Horizon 2020 research and innovation programme (grant agreement No. 819416), and Swarnajayanti Fellowship (No. DST/SJF/MSA01/2017-18). y
		} 
}

\date{}

\begin{document}

\maketitle

\begin{abstract}
%Solution size has been a natural parameter for studying many graph problems in the realm of parameterized algorithms. 
For numerous graph problems in the realm of parameterized algorithms, using the size of a smallest deletion set (called a modulator) into well-understood graph families as parameterization has led to a long and successful line of research.
Recently, however, there has been an extensive study of structural parameters that are potentially much smaller than the modulator size. 
%solution size. 
%; \red{and simultaneously generalize both solution size and width parameters such as treewidth and treedepth.}
 In particular, recent papers [Jansen et al.\ STOC 2021; Agrawal et al.\ SODA 2022] have studied parameterization by the size of the modulator to a graph family $\cH$ ($\mdh(\cdot)$), elimination distance to $\cH$ ($\edh(\cdot)$), and $\cH$-treewidth ($\twh(\cdot)$). These parameters are related by the fact that $\twh$ lower bounds $\edh$, which in turn lower bounds $\mdh$. 
  While these new parameters have been successfully exploited to design fast exact algorithms 
%  (for finding the respective parameters themselves, as well as for other graph problems), 
%  to the best of our knowledge, 
their utility (especially that of $\edh$ and $\twh$)
%  have not been explored for design of
in the context of approximation algorithms is mostly unexplored. 
	
The conceptual contribution of this paper is to present novel algorithmic meta-theorems that expand the impact of these structural parameters to the area of \FPT Approximation, mirroring their utility in the design of exact \FPT algorithms. 
%similar like 
%In particular, we 
%show how 
%these structural parameters 
%are useful for 
%designing exact \FPT algorithms, so are they 
%useful for
%can be used to 
%
Precisely, we show that 
 if a covering or packing problem is definable in Monadic Second Order Logic and has a property called Finite Integer Index (\FII), then the existence of an \FPT Approximation Scheme (\FPTAS, i.e., $(1\pm \epsilon)$-approximation) parameterized by $\mdh(\cdot), \edh(\cdot)$, and $\twh(\cdot)$ is in fact equivalent. As a consequence, we obtain \FPTASes for a wide range of covering, packing, and domination problems on graphs with respect to these parameters. In the process, we show that several graph problems, that are W[1]-hard 
%(or para-NP hard) 
parameterized by $\mdh$, admit \FPTASes not only when parameterized by $\mdh$, but even when parameterized by the potentially much smaller parameter $\twh(\cdot)$.
  In the spirit of  [Agrawal et al.\ SODA 2022], 
%obtain a result, that says that 
%show that
 our algorithmic results highlight a broader connection between these parameters in the world of approximation. 
% Precisely, we show that 
% if a covering or packing problem $\Pi$ is definable in Monadic Second Order Logic and has a property called Finite Integer Index (\FII), then the existence of an \FPTAS parameterized by $\mdh(\cdot), \edh(\cdot)$, and $\twh(\cdot)$ is in fact equivalent. 
%Our results use a simple natural trick, ``bucket versus ocean'', which is based on whether the size of the modulator is comparable to that of the local solution, or it is too small. 
%
As concrete exemplifications of our meta-theorems, we obtain {\FPTAS}es for well-studied graph problems such as \textsc{Vertex Cover, Feedback Vertex Set, Cycle Packing} and \textsc{Dominating Set}, parameterized by $\twh(\cdot)$ (and hence, also by $\mdh(\cdot)$ or $\edh(\cdot)$), where $\cH$ is any family of minor free graphs.

\end{abstract}

\input{introduction}
\input{preliminaries}

\input{deltof}

%\input{mainthorem}

\input{packing}

\input{domset}

\input{connectivity}
%\input{tw-kernel}

%\input{modulator-results}

\input{conclusion}

\bibliography{refs}

\appendix

\input{fii-theorem-proof}

\end{document}

%% file: introduction.tex
%!TEX root = main.tex

\section{Introduction} \label{sec:intro}

One of the most widely studied graph problems in the area of parameterized complexity is the {\sc $\cF$-Vertex Deletion} problem, where the input is a graph $G$ and a number $k$ and the goal is -- ``Compute a set of at most $k$ vertices whose deletion places the resulting graph in the graph family $\cal F$ or correctly conclude that such a set does not exist.'' 
A solution to an instance of {\sc $\cF$-Vertex Deletion} is called a  {\em modulator}  into $\cF$ and there are numerous results in parameterized complexity on exploiting modulators into various graph families to design algorithms. Much of this research has been  motivated by the fact that inputs that have modulators of small size into $\cF$ turn out to be tractable for many problems that are NP-complete in general while being polynomial-time solvable on $\cF$. 
In other words, it is possible to take efficient algorithms for some problems on graphs in $\cF$, and lift them to efficient algorithms for these problems on graphs that are not necessarily in $\cF$, but have a small vertex modulator into $\cF$, i.e., graphs that are ``close'' to $\cF$. 
This leads to fixed-parameter algorithms for these problems (i.e., running time bounded by  $f(k)n^{O(1)}$, where $k$ is the modulator size and $n$ is the input size) under these parameterizations.
 Using the size of the smallest vertex modulator of a graph into tractable graph families or ``the distance  from triviality'' methodology~\cite{GuoHN04} has therefore become   
a rich source of interesting and useful  parameters for graph problems over the last two decades.

In light of the success of this line of research, recent years have seen a shift towards identifying and exploring the power of  ``hybrid'' parameters that are upper bounded by the modulator size as well as certain graph-width measures and can be arbitrarily (and simultaneously) smaller than both the modulator size and these graph-width measures. 
Two specific parameters studied in this line of research are:  $\hh$-elimination distance and  
\hhtwfull{} of~$G$. The $\hh$-elimination distance of  a graph $G$ (denoted $\edh(G)$) was introduced by Bulian and Dawar~\cite{BulianD16} and roughly speaking, it expresses the number of rounds needed to obtain a graph in $\cH$ by removing one vertex from every connected component in each round. We refer the reader to Section~\ref{sec:prelims} for a more formal definition. 
%The reader familiar with the notion of treedepth~\cite{NesetrilM06} will be able to see that this closely follows the recursive definition of treedepth. That is, 
Note that $\edh(G)$ (respectively, $\twh(G)$) can be arbitrarily smaller than both $\hhmd(G)$ and the treedepth of $G$ (respectively, the treewidth of $G$). For example, let $\cH$ be the (infinite) family of complete graphs, and let $G$ be a graph that contains a vertex $v$ that is adjacent to some (non-empty) subset of vertices from a clique of size $t$, for some $t \ge 1$. Then, $\tw(G) = t-1$, whereas $\twh(G) = 0$ (in fact, even $\edh(G) = 0$). Similarly, $\twh(G)$ itself can be arbitrarily smaller than $\edh(G)$ (see  \cite{GanianRS17,JansenK021} for some examples).

Recent work by Agrawal et al.~\cite{AgrawalKLPRSZ22Elimination} and Jansen et al.~\cite{JansenK021} show that for many basic graph problems in the literature and well-understood graph families $\cF$, one can indeed obtain FPT algorithms parameterized by $\edh$ and $\twh$, thus expanding the notion of useful distance from triviality to encompass these parameters as well. Agrawal et al.~\cite{AgrawalKLPRSZ22Elimination} also showed a tight connection between $\mdh, \edh$ and  $\twh$ by showing that for these problems, having an FPT algorithm parameterized by $\mdh$ was sufficient to obtain FPT algorithms parameterized by the other two ``smaller'' parameters. In addition to these results, there is also some recent work on computing $\edh$ and $\twh$, where $\cH$ is the family of bipartite graphs \cite{JansenK21}, or a minor-closed family \cite{MorelleSST23}.

Despite these leaps in our understanding of the parameterized complexity of many problems, limitations remain. For instance, by requiring that the problem be FPT parameterized by $\mdh$, we are implicitly requiring that the problem be polynomial-time solvable on the class $\cH$. This rules out meaningful results for many basic problems and established graph families $\cH$. For instance, it is not interesting to study  {\sc Vertex Cover}  parameterized by $\mdh$ when $\cH$ is the class of planar graphs since {\sc Vertex Cover} is NP-complete on planar graphs~\cite{GareyJ77}. However, it is efficiently approximable on planar graphs (i.e., even has an  Efficient PTAS)~\cite{FominLS18Grid}. This state of the art brings us to the following two natural questions and is the main motivation behind this work.

\begin{tcolorbox}[colback=white!5!white,colframe=gray!75!black]
\begin{quote}
\begin{description}\item[Question 1:] Can good {\em approximation algorithms} for a problem on the class $\cH$ be used to obtain good FPT approximation algorithms for the same problem parameterized by $\mdh$?
\item[Question 2:] Could we obtain positive answers to the above question, but for the parameters $\edh$ and $\twh$?
	
\end{description}

\end{quote}
	
\end{tcolorbox}

The notion of an FPT approximation algorithm is easily motivated by the simultaneous existence of fixed-parameter intractability results  as well as polynomial-time inapproximability results for numerous problems in the literature.
Hence, the topic of  FPT approximation, has been an extremely active area of research in the last decade. For a comprehensive survey of the state of the art, we refer the reader to the survey by Feldman et al.~\cite{FeldmannSLM20}.

We note that we are not the first to study Question 1. In Marx's classic survey~\cite{DBLP:journals/cj/Marx08} on parameterized approximation, he outlines an FPT approximation algorithm for {\sc Chromatic Number} parameterized by the modulator to planar graphs. More recently, Demaine et al.~\cite{DBLP:conf/esa/DemaineGKLLSVP19} studied this question systematically, albeit restricted to polynomial-time approximation. They developed a general theorem that gives sufficient conditions on the problem in order to guarantee an affirmative answer to this question when the class $\cH$ has bounded treewidth or arboricity, but in their context, the approximation ratio of the algorithms depends on $\mdh$. However, Question 1 in the setting of FPT approximation is  largely  unexplored. To  the best of our knowledge, Question 2 has not been considered in the literature and this is the main focus of our paper.

\subsection{Our contributions}

The main conceptual message of the paper is the following meta-result (stated informally):

\begin{tcolorbox}[colback=white!5!white,colframe=gray!75!black]
%\begin{quote}
If a problem can be captured by a very expressible logic fragment (i.e, Counting Monadic Second Order Logic ({\sf CMSO})) and has an appropriate graph replacement subroutine (i.e., Finite Integer Index ({\sf FII})) and fulfils a few other mild requirements, then the existence of an {\FPTAS} for the problem parameterized by any of the three parameters $\mdh, \edh$ and $\twh$ is equivalent to each other.

%\end{quote}
	
\end{tcolorbox}

The formal version of our first meta-theorem in the context of the well-studied family of vertex-deletion problems is given below. Say that a family of graphs is {\em well-behaved} if it is hereditary and  closed under disjoint union.

\begin{restatable}{theorem}{monotoneFIIFptas} \label{thm:monotone-fii-fptas}
Let $\cH, \cF$ be well-behaved families of graphs, where $\cF$ is CMSO-definable. Suppose $\Pi = \deltof$ has \FII.  
	Then, the following statements are equivalent.
	\begin{enumerate}
		\item $\Pi$ admits an \FPTAS parameterized by $\mdh(\cdot)$ and $\epsilon$.
		\item $\Pi$ admits an \FPTAS parameterized by $\edh(\cdot)$ and $\epsilon$.
		\item $\Pi$ admits an \FPTAS parameterized by $\twh(\cdot)$ and $\epsilon$.
	\end{enumerate}
\end{restatable}

We also prove meta-theorems similar to Theorem~\ref{thm:monotone-fii-fptas} where $\Pi$ is {\sc ${\cal F}$-Subgraph Packing} (i.e., pack maximum number of vertex-disjoint subgraphs isomorphic to graphs in $\cF$) or {\fpacking} (i.e., pack maximum number of vertex-disjoint minor-models of graphs in $\cF$).

In order to invoke these equivalence theorems,  we first show that a wide range of graph problems, that are W[1]-hard (or para-NP hard) parameterized by $\mdh$, admit \FPT Approximation Schemes ({\FPTAS}es, i.e., $(1\pm \epsilon)$-approximation) when parameterized by  $\mdh$. Hence, as corollaries of our meta-results, we also get the first {\FPTAS}es for the following (non-exhaustive list of)  problems parameterized by $\twh$ (and hence, also by $\edh$). When $\cH$ is apex-minor free, we obtain this result for  {\indset}, {\sc Triangle Packing}. When $\cH$ is minor free, we are able to infer the same result for {\sc Vertex Cover}, {\sc Feedback Vertex Set}, {\sc Cycle Packing} (and more generally, {\fpacking}).

Finally, we identify {\sc Dominating Set} and its connectivity constrained variant along with {\sc Connected Vertex Cover}   as interesting special cases that are not covered by the framework developed above and give  purpose-built {\FPTAS}es for them parameterized by $\mdh$ and then also prove an equivalence theorem in the same spirit as Theorem~\ref{thm:monotone-fii-fptas}. 
%\todo{add dom set theorem? What about CVC and CDS?}

In summary, our work highlights many natural problems for which the answer to Question 2 is affirmative. Finally, we note that in \Cref{thm:monotone-fii-fptas} and its variants, due to their generality, we only obtain \emph{non-uniform} \FPTASes. In Conclusion, we briefly discuss certain scenarios where it may be possible to obtain uniform \FPTASes. 

\subsection{Our techniques}
We first summarize the technique behind our {\FPTAS}es parameterized by $\mdh$, which we dub ``Bucket vs Ocean''. Recall that this is required in order to invoke Theorem~\ref{thm:monotone-fii-fptas} to get our eventual results, i.e., {\FPTAS}es parameterized by $\twh$ and $\edh$. We remark that throughout the paper, we assume that our \FPTAS parameterized by, say $\mdh$ (resp.\ $\edh, \twh$) is also provided with a modulator to $\cH$ (resp.\ $\cH$-elimination/tree decomposition) of the appropriate size (resp.\ width). For $\edh$ and $\twh$, one can use algorithms from \cite{GanianRS17,JansenK021,AgrawalKLPRSZ22Elimination,JansenK21} to compute the appropriate decompositions exactly. Alternatively, one can use the recent result of \cite{JansenK023} to compute constant approximations thereof. 

Consider a special case of   \deltof, say {\sc Vertex Cover}. Consider an $H$-minor free graph family $\cH$, where VC admits an \EPTAS. Consider a graph $G$ and $M \subseteq V(G)$ of size $\mdh(G)$ such that $G - M \in \cH$. We use the known \EPTAS on $G - M$ to compute an $(1+\epsilon/2)$-approximate solution to VC, call it $S$. Next, we compare $|S|$ with $|M|$, and consider two cases. If $|M| < \epsilon/3 \cdot |S| \le (1+\epsilon/2) \cdot \OPT(G-M)$, i.e., $\OPT(G-M)$ is like an ``ocean'', compared to a ``bucket'' of water that is $M$, then we can add the bucket to the ocean ``for free'', i.e., $|S| + |M| \le (1+\epsilon) \cdot \OPT(G-M) \le (1+\epsilon)\cdot \OPT(G)$. Otherwise, if $|M|$ and $|S|$ are comparable, then $\OPT(G) \le 3 \cdot \mdh(G)/\epsilon$, in which case we can use the \FPT algorithm for \textsc{Vertex Cover}, which is in fact in \FPT in $\mdh$ and $\epsilon$. This is the high level idea behind our {\FPTAS}es parameterized by $\mdh$, although one has to overcome several problem-specific challenges to make it work for the other problems.

A more sophisticated version of the idea also turns out to be useful in proving the equivalence theorem. Again, let us consider the example of {\sc Vertex Cover} and $\cH$ being an $H$-minor free graph family. Having armed ourselves with an \FPTAS parameterized by $\mdh$ as in the previous paragraph, now our task is to generalize to $\twh$. To this end, we consider each bag $\chi(t)$ corresponding to a node $t$ in the $\cH$-tree decomposition (defined formally in the next section), which consists of $\ell \le \twh(G)+1$ vertices, say $R_t$, which locally act like a modulator to a disjoint ``base graph'', say $G[H_t] \in \cH$. We use the same ``Bucket vs Ocean'' idea to classify each node as \emph{good} or \emph{bad}. A node $t$ is \emph{good}, if a $(1+\epsilon)$-approximate solution $S_t$ for $G[H_t]$ is like an ``ocean'' compared to $R_t$. In this case, we can again add $R_t$ to $S_t$ ``for free''. Otherwise, we say that a node $t$ is bad, if $\OPT(G[H_t])$ is bounded by $3\ell/\epsilon$. From here, our task is to reduce the treewidth of the graph induced by the vertices in the bags of bad nodes. To this end, we use the \FII property as well as the \FPTAS to perform a series of ``graph replacements'' in each bad node, where in each iteration, we replace the corresponding bag with an ``equivalent subgraph'' of size bounded by a function of $\ell$ and $\epsilon$. At the end of this procedure, the resulting graph has treewidth, and we can use the CMSO-definability and Courcelle's theorem \cite{Courcelle90} to find an optimal solution (that can then be translated back to the original instance). It can be shown that the resulting solution is a $(1+\epsilon)$-approximate solution for the entire graph, and the whole algorithm runs in time $\FPT$ in $\twh$ and $\epsilon$. 

%% file: preliminaries.tex
\section{Preliminaries}\label{sec:prelims}
For an instance $\cI$ of an optimization problem $\Pi$, we denote the value/size of an optimal solution by $\OPT_\Pi(\cI)$, and may omit the subscript if the problem is clear from the context. An algorithm that, for any feasible instance of minimization (resp.\ maximization) problem $\Pi$, returns a solution of cost/size at most (resp.\ at least) $\alpha \cdot \OPT$, is called an $\alpha$-approximation. We say that a minimization (resp.\ maximization) problem $\Pi$ admits an \FPTAS (FPT-Approximation Scheme) parameterized by a parameter $t$ and $\epsilon$, if there exists an $(1+\epsilon)$-approximation (resp.\ $(1-\epsilon)$-approximation) algorithm with running time of the form $f(t, \epsilon) \cdot |\cI|^{\Oh(1)}$, for some computable function $f$. Similarly, we say that $\Pi$ admits an \EPTAS (efficient polynomial-time approximation scheme), if there exists a $(1+\epsilon)$-approximation (resp.\ a $(1-\epsilon)$-approximation) algorithm with running time of the form $g(\epsilon) \cdot |\cI|^{\Oh(1)}$, for some computable function $g$.

\smallskip
\noindent 
{\bf Graphs.} Consider a graph $G$. We denote by $V(G)$ and $E(G)$, the set of vertices and the set of edges of the graph $G$, respectively. We say that a graph $H$ is a subgraph of $G$ if $V(H)\subseteq V(G)$ and $E(H)\subseteq V(G)$. We say that $H$ is an induced subgraph of $G$ if $E(H)=\{uv~|~u,v\in V(H), uv\in E(G)\}$. Let $S\subseteq V(G)$, by $G[S]$ we denote the graph $G$ induced on set $S$. For $S\subseteq V(G)$, by $G-S$, we denote the graph $G[V(G)\setminus S]$. That is, the vertex set of $G-S$ is $V(G)\setminus S$ and the edge set is $\{uv|u,v\notin S, uv\in E(G)\}$. Let $F\subseteq E(G)$, by $G-F$, we denote the graph with vertex set $V(G)$ and edge set $E(G)\setminus F$. By {\em vertex (edge) deletion} operation or deletion of a vertex $w$ (edge $uv$) in $G$, we  mean the graph $G-w$ ($G-\{uv\}$) obtained by deleting vertex $w$ (edge $uv$) from $G$. For singleton vertex and edge subsets, we use $G - w$ and $G - uv$. By {\em edge contraction} operation or contraction of an edge $uv$ in $G$ we mean deleting vertices $u,v$ in $G$ and adding a new vertex $w$, and making $w$ adjacent to all the neighbors of $u,v$, i.e., the resulting graph has vertex set $V(G)\setminus \{u,v\}\cup \{w\}$ and edge set $\{xy|x,y\notin \{u,v\}\} \cup \{wx|x\in N(u)\cup N(v)\}$. Let $G_1$ and $G_2$ be two graphs on $n_1,n_2$ vertices and $m_1,m_2$ edges, respectively. By disjoint union of $G_1,G_2$ denoted as $G_1\uplus G_2$, we denote the graph with $n_1+n_2$ vertices and $m_1+m_2$ edges, where vertex set of $G_1\uplus G_2$ is $V(G_1)\cup V(G_2)$ and edge set of $G_1\uplus G_2$ is $E(G_1)\cup E(G_2)$. We refer the reader to the book of Diestel \cite{diestel-book}  for standard graph-related terms that are not explicitly defined here.

%\todo[inline]{Basics on graphs: vertex deletion, subgraphs, minors, disjoint unions.}

\begin{definition}[Graph Minor]
	A graph $H$ is a \emph{minor} of $G$ if it can be obtained by a sequence of vertex deletions, edge deletions, and edge contractions.
\end{definition}

Let $\cG$ denote the family of all graphs. We say that a family $\cF \subseteq \cG$ of graphs is \emph{hereditary} if for any graph $G \in \cF$, every induced subgraph $G'$ of $G$ also belongs to $\cF$. We say that a family $\cF \subseteq \cG$ is \emph{closed under disjoint union} if for any $G_1, G_2 \in \cF$, the graph $G$ obtained by taking the disjoint union of $G_1$ and $G_2$, also belongs to $\cF$. We say that a family $\cF \subseteq \cG$ is \emph{well-behaved} if it is (i) hereditary, and (ii) closed under disjoint union. 

%\todo[inline]{Tree decompositions, treewidth, modulator, elimination distance, treewidth wrt $\cH$.}

\smallskip
\noindent 
{\bf Graph Decomposition.} 
The following definitions are borrowed from \cite{AgrawalKLPRSZ22Elimination}. For a graph $G$, $\mdh(G)$ denotes the size of a smallest vertex set $S$ such that $G-S\in \cH$. If $G-S\in \cH$, then $S$ is called a {\em modulator} to $\cH$.

\begin{definition} \label{def:h-elimination-decomp}
	For a graph family $\cH$, an \emph{$\cH$-elimination decomposition} of $G$ is a triple $(T, \chi, L)$, where $T$ is a rooted forest, $\chi: V(T) \to 2^{V(G)}$, and $L \subseteq V(G)$ such that:
	\begin{enumerate}
		\item For each internal node $t \in V(T)$, we have that $|\chi(t)| \le 1$, and $\chi(t) \subseteq V(G) \setminus L$.
		\item The sets $(\chi(t))_{t \in V(T)}$ form a partition of $V(G)$.
		\item For each leaf $t$ in $T$, we have $\chi(t) \subseteq L$, such that the graph $G[\chi(t)]$, called a base component, belongs to $\cH$. Furthermore, $(\chi(t))_{\text{leaf } t}$ forms a partition of $L$.
		\item For each edge $uv \in V(G)$, if $u \in \chi(t_1)$, and $v \in \chi(t_2)$, then $t_1$ and $t_2$ are in ancestor-descendant relation in $T$.
	\end{enumerate}
\end{definition}

The depth of a rooted tree $T$ is the maximum number of edges on a root-to-leaf path in $T$. We refer to the union of base components as the set of base vertices. The $\cH$-elimination distance of $G$, denoted as $\edh(G)$, is the minimum depth of an $\cH$-elimination forest for $G$. Note that a pair $(T, \chi)$ is a (standard) elimination forest of $\cH$ is a class of empty graphs, i.e., the base components are empty. In this case, $\edh(G)$ is known as the \emph{treedepth} of $G$, and is denoted as $\td(G)$.

Just like the notion of $\cH$-elimination decomposition generalizes the notion of elimination forest, the following is an analogous generalization of the notion of tree decomposition.

\begin{definition} \label{def:h-treewidth}
	For a graph family $\cH$, an \emph{$\cH$-tree decomposition} of a graph $G$ is a triple $(T, \chi, L)$, where $T$ is a rooted tree, $\chi: V(T) \to 2^{V(G)}$, and $L \subseteq V(G)$, such that:
	\begin{enumerate}
		\item For each $v \in V(G)$, the nodes $\LR{t \in V(T) : v \in \chi(t)}$ induce a non-empty connected subtree in $T$.
		\item For each edge $uv \in E(G)$, there is a node $t \in V(G)$ with $\{u, v\} \in \chi(t)$.
		\item For each vertex $v \in L$, there is a \emph{unique leaf} $t \in V(T)$ for which $v \in \chi(t)$. 
		\item For each leaf node $t \in V(T)$, the graph $G[\chi(t) \cap L] \in \cH$.
	\end{enumerate}
\end{definition}

The \emph{width} of an $\cH$-tree decomposition is defined as $\max\LR{ 0, \max_{t \in V(T)} |\chi(t) \setminus L| - 1 }$. The $\cH$-treewidth of $G$, denoted by $\twh(G)$, is the minimum width of an $\cH$-tree decomposition of $G$. The connected components of $G[L]$ are called base components, and the vertices in $L$ are called base vertices. 

A pair $(T, \chi)$ is a (standard) \emph{tree decomposition} if $(T, \chi, \emptyset)$ satisfies all conditions of an $\cH$-decomposition, where the choice of $\cH$ is irrelevant.

\medskip

\noindent 
{\bf Counting Monadic Second Order Logic.}
\label{sec:prelimsCMSO}
The syntax of Monadic Second Order Logic (MSO) of graphs includes the logical connectives $\vee,$ $\land,$ $\neg,$ 
$\Leftrightarrow,$ $\Rightarrow,$ variables for vertices, edges, sets of vertices and sets of edges, the quantifiers $\forall$ and $\exists$, which can be applied to these variables, and five binary relations: 
\begin{enumerate}
\item $u\in U$, where $u$ is a vertex variable and $U$ is a vertex set variable; 
\item $d \in D$, where $d$ is an edge variable and $D$ is an edge set variable;
\item $\mathbf{inc}(d,u),$ where $d$ is an edge variable, $u$ is a vertex variable, and the interpretation  is that the edge $d$ is incident to  $u$; 
\item $\mathbf{adj}(u,v),$ where $u$ and $v$ are vertex variables, and the interpretation is that $u$ and $v$ are adjacent;
\item  equality of variables representing vertices, edges, vertex sets and edge sets.
\end{enumerate}

Counting Monadic Second Order Logic (CMSO) extends MSO by including atomic sentences testing whether the cardinality of a set is equal to $q$ modulo $r,$ where $q$ and $r$ are integers such that $ 0\leq q<r $ and $r\geq 2$. That is, CMSO is MSO with the following atomic sentence: 
$\mathbf{card}_{q,r}(S) = \mathbf{true}$ if and only if $|S| \equiv q \pmod r$, where $S$ is a set.
We refer to~\cite{ArnborgLS91,Courcelle90} for a detailed introduction to  CMSO. 
%~\cite{ArnborgLS91,Courcelle90,courcelle1997expression}

%% file: deltof.tex
%!TEX root = main.tex

\section{Vertex Deletion to $\cF$} \label{sec:deltof}
In this section we give our ``first equivalence'' result. Towards that, first in \Cref{subsec:deltof-examples}, we design concrete \FPTASes for several \deltof type problems (e.g., \textsc{Vertex Cover, Feedback Vertex Set} parameterized by $\mdh$, where $\cH$ is an apex-minor family of graphs.  
Then, in  \Cref{subsec:deltof-equiv} we prove our main equivalence theorem, which lets us extend the previous \FPTASes parameterized by $\mdh(\cdot)$ to $\edh(\cdot)$, and $\twh(\cdot)$. 

\subsection{Preliminaries for \deltof} \label{subsec:deltof-prelims}

We focus on the following problem, specifically on the case where $\cF$ is some fixed well-behaved family of graphs. 

\begin{tcolorbox}[colback=white!5!white,colframe=gray!75!black]
	\deltof 
	\\\textbf{Input:} An instance $(G, k)$, where $G = (V, E)$ is a graph, and $k \ge 0$ is an integer.
	\\\textbf{Question:} Does there exist a subset $S \subseteq V(G)$ of size at most $k$, such that $G - S \in \cF$?
\end{tcolorbox}
\medskip\noindent\textbf{\textsf{Optimization variant.}} Let $\cF$ be a well-behaved family of graphs. For a graph $G$, a set $S \subseteq V(G)$ is said to be a \emph{solution} to \deltof, if $G - S \in \cF$. In the optimization variant of \deltof, we want to find a solution of the smallest cardinality. By slightly abusing the notation, we will use \deltof to refer to the decision as well as the optimization version, and will only disambiguate when strictly necessary. Note that assuming $\cF$ contains at least one graph, the hereditary property implies that the empty graph belongs to $\cF$, which means that $V(G)$ is always a solution to \deltof. 

\begin{observation} \label{obs:monotone}
	Let $\cF$ be a well-behaved family, and let $\Pi = \deltof$. Then, for any graph $G$ and any subset $S \subseteq V(G)$, it holds that 
	$$\OPT_\Pi(G - S) \le \OPT_\Pi(G) \le \OPT_\Pi(G - S) + |S|.$$
\end{observation}
\begin{proof}
	Let $T \subseteq V(G)$ (resp.\ $T' \subseteq V(G) \setminus S$) be an optimal solution to $\Pi$ on $G$ (resp.\ $G - S$). Let $T'' = T \cap (V(G) \setminus S)$, and note that $G'' \coloneqq (G - S) - T''$ is an induced subgraph of $G - T \in \cF$. Since $\cF$ is hereditary, $G'' \in \cF$, which implies that $T''$ is a solution to $\Pi$ on $G''$. Then, $\OPT_{\Pi}(G - S) = |T'| \le |T''| \le |T| = \OPT_{\Pi}(G)$, the first inequality follows.
	
	For the second inequality, note that $(G - S) - T' \in \cF \iff G - (S \cup T') \in \cF$. That is, $S \cup T'$ is a feasible solution to $\Pi$ on $G$. Thus, $\OPT_\Pi(G) \le |S \cup T'| = |S| + |T'| \le |S| + \OPT_\Pi(G - S)$.
\end{proof}

\begin{definition}[Self-reducibility] \label{def:self-reducibility}
	Let $\Pi$ be a decision problem on graphs, and let $\cA$ be an algorithm for $\Pi$. We say that $\Pi$ is \emph{self-reducible} if, there exists an algorithm $\cA'$, that makes at most polynomially many calls to $\cA'$, and at most a polynomial overhead, and has the following behavior. Given a no-instance of $\Pi$, $\cA'$ correctly concludes that it is a no-instance; otherwise, given a yes-instance of $\Pi$, it outputs a ``solution'' witnessing the yes-instance.
\end{definition}
The preceding definition is imprecise, since the notion of ``solution'' is not formally defined. However, for all the problems that we will consider in this paper, the notion of solution will be obvious. For example, the decision version of \vc (resp.\ \indset) takes input a pair $(G, k)$, and asks whether there exists a vertex cover (independent set) of size at most (at least) $k$. Here, the solution is naturally a vertex cover (independent set) of the prescribed size.

\begin{lemma} \label{lem:self-reducibility}
	For any fixed family $\cF \subseteq \cG$, the problem \deltof is self-reducible. %That is, let $\cA$ be an algorithm that, that correctly decides whether a given instance $(G, k)$ is a yes-instance of \deltof in time $T$, where $n = |V(G)|$. Then, there exists an algorithm $\cA'$, that, given an instance $(G, k)$, runs in time $T \cdot (n+1) + \polyn$, and either finds a solution $S \subseteq V(G)$ of size at most $k$, or correctly concludes that there exists no solution of size at most $k$.
\end{lemma}
\begin{proof}
	Consider the input $(G, k)$ to $\cA'$. We first call $\cA$ on the instance $\cI = (G, k)$, which takes time $T$. If $\cA$ reports that $\cI$ is a no-instance of \deltof, then so does $\cA'$. Otherwise, assume that $\cI$ is a yes-instance. Thus, there exists $S \subseteq V(G)$ of size at most $k$ such that $G - S \in \cF$. For any $v \in V(G)$, define an instance $\cI_v \coloneqq (G - v, k-1)$.
	
	For any $v \in S$, observe that the graph $(G - v) - (S \setminus \{v\}) = G - S \in \cF$. Therefore, for any $v \in S$, $\cI_v$ is a yes-instance. Furthermore, consider a $w \in V(G)$ such that $\cI_w$ is a yes-instance. Then, there exists a subset $S' \subseteq V(G) \setminus \{w\}$ of size at most $k-1$ such that $(G - w) - S' \in \cH  \implies G - (S' \cup \{w\}) \in \cF$. Thus, $S' \cup \{w\}$ is a solution of size at most $k$.
	
	So $\cA'$ iterates over the vertices $v \in V(G)$, and use the algorithm $\cA$ to decide the instance $\cI_v$. The first time $\cA$ reports that $\cI_v$ is a yes-instance, we add $v$ to our solution, and continue the procedure on the graph obtained by deleting $v$ from the current graph, and reducing $k$ by $1$. We stop when either the graph becomes empty or $k$ becomes $0$. The correctness of the algorithm follows from the previous claim. As for the running time, $\cA$ is called initially on $(G, k)$, and then at most once for each instance $\cI_v$ for each $v \in V(G)$. Accounting for the polynomial time taken for deleting vertices and incident edges, the claimed running time follows.
\end{proof}

\input{deltof-examples}

\input{maintheorem}

%% file: deltof-examples.tex
\subsection{FPT-ASes for \deltof Parameterized by $\mdh(\cdot)$} \label{subsec:deltof-examples}

Let $\cH, \cF$ be well-behaved families of graphs, such that $\Pi = \deltof$ admits an \EPTAS for (the optimization variant of) $\Pi$ on any $G \in \cH$. Then, we show how to design an \FPTAS for $\Pi$ parameterized by $\mdh(\cdot)$, and $\epsilon$. Note that, for a graph $G$, when $p \coloneqq \mdh(G) = 0$, the graph already belongs to $\cH$, and in this case our algorithm is essentially an \EPTAS. In this sense, we generalize the assumed \EPTAS on $\cH$ to a larger family of graphs that are at most $p$ vertices away from $\cH$. 

Let $M \subseteq V(G)$ be a modulator to $\cH$ of size $p$, i.e., $G - M =: F \in \cH$. We first compute a $(1+\epsilon/2)$-approximate solution $S$ for $F$ using \EPTAS in time $g(\epsilon) \cdot \polyn$. Note that $S \subseteq V(G) \setminus M$, and $p \le (1+\epsilon/2) \cdot \OPT(F) \le (1+\epsilon/2) \OPT(G)$, where the last inequality follows from \Cref{obs:monotone}. 
%\iffalse
%\begin{lemma}
%	For any $X \subseteq V(G), \ \text{OPT}_{G \setminus X} \leq \text{OPT}_{G} \leq \text{OPT}_{G \setminus X} + |X|$
%\end{lemma}
%
%\begin{proof}
%	
%	Suppose, $ S = \text{Sol}_{G} \setminus X$. Hence, $|S| \leq \text{OPT}_{G}$. Observe that $S$ is a solution for $G \setminus X$, otherwise it contradicts that $\text{Sol}_{G}$ is a valid solution for $G$. So, $\text{OPT}_{G \setminus X}$ for $G \setminus X$ is upper bounded by $|S|$. So, $\text{OPT}_{G \setminus X} \leq |S| \leq \text{OPT}_{G}$.
%	
%	Observe that $X \cup \text{Sol}_{G \setminus X}$ is a valid solution for $G$ as $\text{Sol}_{G \setminus X}$ is a valid solution for $G \setminus X$ and $X$ takes care of edges inside $G[X]$ as well as edges between $X$ and $V(G)\setminus X$. Hence size of optimum solution for $G$ is upper bounded by  $X \cup \text{Sol}_{G \setminus X}$. Hence, $\text{OPT}_{G} \leq \text{OPT}_{G \setminus X} + |X|$. 
%\end{proof}
%
%
%From EPTAS on $G \setminus X$, we have that $\text{Sol}_{G \setminus X} \leq \text{S}_{G \setminus X} \leq (1 + \epsilon /2). \text{OPT}_{G \setminus X} \leq (1 + \epsilon /2). \text{OPT}_{G}$.
%\fi 
Depending on the relative sizes of $M$, and the approximate solution $S$, we consider the following two cases. 
\begin{enumerate}
	\item $p \leq \frac{\epsilon}{3} \cdot |S|$. 
	\\Then we simply return $X \cup S$ as a solution for $G$. Note that $|M| + |S| \le (1+\epsilon/3) (1+\epsilon/2) \cdot \OPT(F) \le (1+\epsilon) \cdot \OPT(F) \le (1+\epsilon) \cdot \OPT(G)$.
	
	\item $p > \frac{\epsilon}{3} \cdot |S|$, i.e, $\OPT(F) \le |S| \le \frac{3p}{\epsilon}$. 
	\\Then, by \Cref{obs:monotone}, $\OPT(G) \le \OPT(F) + |M| \le p + \frac{3p}{\epsilon} =: p'$. In this case, we use an FPT algorithm parameterized by the solution size, to find an optimal solution for $G$. Note that this takes $h(p') \cdot \polyn$, i.e., $f(p, \epsilon) \cdot \polyn$ time, where $f(p, \epsilon) = h(p+3p/\epsilon)$. 
\end{enumerate}

Thus, we get the following theorem.

\begin{theorem}
\label{thm:fptasvdMod}
	Let $\cF, \cH$ be well-behaved families of graphs. Moreover, suppose that $\Pi = \deltof$ admits an \EPTAS on any graph in $\cH$, and that $\Pi$ is \FPT parameterized by the solution size. Then, there exists an \FPTAS for $\Pi$, parameterized by $p \coloneqq \mdh(G)$ with running time $\max\{g(\epsilon), f(p, \epsilon)\} \cdot \polyn$. 
\end{theorem}

%\iffalse
%
%\begin{proof}
%	For the first case, the algorithm returns $X \cup \text{S}_{G \setminus X}$ in time $2^{\mathcal{O}(1/ \epsilon)}$.
%	
%	\begin{align*}
%		|X \cup \text{S}_{G \setminus X}|  &\leq (\epsilon /2) \ \text{OPT}_{G} + (1 + \epsilon /2). \text{OPT}_{G} 
%		\\& =(1 + \epsilon). \text{OPT}_{G}
%	\end{align*}
%	
%	For the second case, the algorithm returns an optimum solution in time $2^{\mathcal{O}(k/ \epsilon)} \cdot n^{\mathcal{O}(1)}$. 
%	
%	So overall, the algorithm returns $(1 + \epsilon)$-approximate solution in $2^{\mathcal{O}(k/ \epsilon)} \cdot n^{\mathcal{O}(1)}$ time.
%\end{proof}%!TEX root = main.tex
%\fi

\medskip\noindent\textbf{\textsf{Corollaries.}} Let $\Pi$ be either \textsc{Vertex Cover} or \textsc{Feedback Vertex Set}. Note that $\Pi$ is \deltof, where $\cF$ is family of isolated vertices, and forests, respectively. If $\cH$ is a family of apex-minor free graphs, then $\Pi$ admits an \EPTAS on $\cH$ with running time $2^{\Oh(1/\epsilon)} \cdot \polyn$ (\cite{FominLS18Grid}, Corollary 2). Furthermore, $\Pi$ is \FPT parameterized by solution size on general graphs, with running time $2^{\Oh(k)} \cdot \polyn$, where $k$ denotes the solution size. Therefore, we get an \FPTAS parameterized by $p \coloneqq \mdh(G)$ with running time $2^{\Oh(p/\epsilon)} \cdot \polyn$ for \textsc{Vertex Cover} and \textsc{Feedback Vertex Set}.

We note that it is possible to improve the running time of \FPTAS for \textsc{Vertex Cover} using an even simpler algorithm -- we simply guess the intersection of an optimal solution with $X$ by iterating over all subsets of $X$. Let $Y \subseteq X$ be a guess. If $Y$ is feasible, then it must be that $U \coloneqq X \setminus Y$ must be an independent set, and $M \coloneqq N(U) \cap (V(G) \setminus X)$ must belong to the solution. Thus, the remaining graph is $G \setminus (M \cup X)$, which belongs to the family $\cH$. Then, we use the $2^{\Oh(1/\epsilon)} \cdot \polyn$ time \EPTAS on $\cH$ to obtain a solution $S$ and return the smallest possible solution of the form $Y \cup M \cup S$. Thus, we get a $2^{\Oh(\mdh(G) + 1/\epsilon)} \cdot \polyn$ time \FPTAS. 

\begin{corollary}
	There exists an FPT-AS for \textup{\textsc{Vertex Cover}} (resp.\ \textup{\textsc{Feedback Vertex Set}}) that runs in time $2^{\Oh(p + \frac{1}{\epsilon})} \cdot \polyn$ (resp.\ $2^{\Oh(p/\epsilon)} \cdot \polyn$), where $p = \mdh(G)$, and $\cH$ is a family of apex-minor free graphs. 
\end{corollary}

More generally, this result extends for \textsc{Planar $\mathcal{R}$-Deletion}, which is a special case of $\deltof$, where the family $\cF$ is characterized by a set $\mathcal{R}$ of forbidden minors containing at least one planar graph. We note that Fomin et al.\ \cite{FominLS18Grid} give an \EPTAS for this problem on $H$-minor free graphs. On the other hand, this problem is known to admit a single-exponential \FPT algorithm parameterized by the solution size via the results from Fomin et al.\ \cite{FominLMS12Deletion}. Thus, we get the following corollary.

\begin{corollary}
\label{cor:planarRdelmod}
	There exists an \FPTAS for \textup{\textsc{Planar $\mathcal{R}$-Deletion}}, that runs in time $f(p, \epsilon) \cdot \polyn$, where $p = \mdh(G)$, and $\cH$ is a family of $H$-minor free graphs, for some fixed graph $H$ and some computable function $f$. In particular, this implies an \FPTAS for \textup{\textsc{Treewidth-$\eta$-Modulator}}.\footnote{For a fixed $\eta \ge 0$, \textsc{Treewidth-$\eta$-Modulator} is the problem of deciding whether one can delete at most $k$ vertices from a given graph $G$, such that the resulting graph has treewidth at most $\eta$.}
\end{corollary}

Finally, we mention \textsc{Odd Cycle Transversal}. Note that there exists a polynomial time $\frac{9}{4}$-approximation for \textsc{Odd Cycle Transversal} on planar graphs \cite{GoemansW98}, and there exists a $3^k \polyn$ time algorithm for the problem parameterized by the solution size \cite{ReedSV04}. Thus, by using the $9/4$-approximation instead of an $(1+\epsilon/2)$-approximation on $G - X$, followed by a similar case analysis, we get the following result.

\begin{corollary} \label{cor:oct}
	There exists an algorithm that runs in time $2^{\Oh(p/\epsilon)} \cdot \polyn$, where $p$ is the size of planar vertex deletion set, to compute a $(\frac{9}{4} + \epsilon)$-approximation for \textup{\textsc{Odd Cycle Transversal}}.
\end{corollary}

%\medskip\noindent\textbf{\textsf{Lossy Kernels.}} Note that it is also possible to extend the same idea to obtain lossy kernels for these problems parameterized by $\mdh()$. In case 1, the reduction algorithm outputs a trivial instance. The solution lifting algorithm discards the solution given for the trivial instance, and returns $X \cup S$, which is a $(1+\epsilon)$-approximation. In the second case, since the size of the solution for the graph is bounded, we simply return use a standard polynomial kernels for \textsc{Vertex Cover} (resp.\ \textsc{Feedback Vertex Set}). \red{Why does this fit in the lossy kernel framework? I am not sure if we can map an approximate solution of the kernelized instance back to the original instance...}

%% file: maintheorem.tex
%!TEX root = main.tex
\subsection{Equivalence Theorem} \label{subsec:deltof-equiv}
%\todo[inline]{need to reorganize the lemma structure internally, but the proof should be okay}
In this section we ``extend'' the result of Theorem~\ref{thm:fptasvdMod} ``using Theorem~\ref{thm:fptasvdMod}''.  In particular, we get the following equivalent statement. 

%\begin{theorem} \label{thm:monotone-fii-fptas}
%Let $\cH, \cF$ be well-behaved families of graphs. Further, let $\Pi = \deltof$ be CMSO-definable and has \FII.  
%	Then, the following statements are equivalent.
%	\begin{enumerate}
%		\item $\Pi$ admits an \FPTAS parameterized by $\mdh(\cdot)$ and $\epsilon$.
%		\item $\Pi$ admits an \FPTAS parameterized by $\edh(\cdot)$ and $\epsilon$.
%		\item $\Pi$ admits an \FPTAS parameterized by $\twh(\cdot)$ and $\epsilon$.
%	\end{enumerate}
%	
%	%Then, there is an algorithm that runs in $g(\ell, \epsilon) \cdot n^{\Oh(1)}$-time algorithm that returns a $(1+\epsilon)$-approximation to $\Pi$ on instance $(G, k)$, where $\ell = \edh(G)$ (resp.\ $\twh(G)$).
%\end{theorem}
\monotoneFIIFptas*
\begin{proof}
	To prove that the three items are equivalent, we first note that for any graph $G$ and any well-behaved family $\cH$, it holds that $\mdh(G) \ge \edh(G) \ge \twh(G)$. Therefore, $3 \implies 2 \implies 1$. Thus, we focus on proving $1 \implies 3$. Let \pmd denote the assumed \FPTAS for $\Pi$ parameterized by $\mdh$ and $\epsilon$. That is, for any graph $G$ on $n$ vertices, $\pmd$ runs in time $f(\mdh(G), \epsilon) \cdot \polyn$, and returns a subset $S \subseteq V(G)$ such that (i) $G - S \in \cF$, and (ii) $|S| \le (1+\epsilon) \cdot \OPT_\Pi(G)$. \footnote{Note that \pmd may in fact be a family of approximation algorithms, containing an algorithm for each value of $\mdh(\cdot)$ and $\epsilon$. We slightly abuse the notation and assume that we use an appropriate algorithm from this family, based on the value of $\mdh$ of the relevant graph, and $\epsilon$.}. 
	
	Let $G = (V, E)$ be a graph, and let $(T, \chi, L)$ be the given $\cH$-tree decomposition of $G$ of width at most $\twh(G)$, and let $\ell \coloneqq \twh(G)$ (resp. $\ell \coloneqq \edh(G)$). Recall that $\mathcal{P}_\md$ returns a $(1+\epsilon)$-approximation to $\Pi$ on any instance $(G, k)$ of $\Pi$.
	
	For every leaf node $t \in V(T)$, let $V_t \coloneqq \chi(t)$, and let $H_t = V_t \cap L$ and $R_t = V_t \setminus L$. The properties of $\cH$-tree decomposition imply that $G[H_t] \in \cH$, and $|R_t| \le \ell$. Equivalently, $R_t \subseteq V(G) \setminus L$ is a modulator to $\cH$ for the graph $G[\chi(t)]$.
	Note that the base vertices, i.e., the vertices in $L$, are partitioned across the leaf bags, which implies that the graphs $\{G[H_t]\}_{\text{leaf\ } t \in V(T)}$ are disjoint.
	
	We iterate over every leaf nodes $t \in V(T)$, and proceed as follows. Note that $G[H_t] \in \cH$, thus $G[H_t]$ has a modulator of size zero to the family $\cH$. Therefore, we can use $\mathcal{P}_\md$ to obtain a set $S_t \subseteq H_t$ that is a $(1+\epsilon/2)$-approximation for $\Pi$ on the graph $G[H_t]$, i.e., $|S_t| \le (1+\epsilon/2) \cdot \OPT(G[H_t])$. This takes $f(0, \epsilon/2) \cdot |H_t|^{\Oh(1)} \le h(\epsilon) \cdot n^{\Oh(1)}$ time for some function $h(\cdot)$. In other words, $\mathcal{P}_\md$ is an \EPTAS for $\Pi$ on $G[H_t]$. 
	
	For every leaf node $t \in V(T)$, if $|R_t| \le \epsilon/3 \cdot |S_t|$, then we say that the node is \emph{good}. Otherwise, we say that $t \in V(T)$ is \emph{bad}. Furthermore, all non-leaf nodes are classified as \emph{bad}. Note that for a good node $t \in V(T)$, $|R_t| \le (1+\epsilon/2) \cdot (\epsilon/3) \cdot \OPT(G[H_t]) \le (\epsilon/2) \cdot \OPT(G[H_t]) $. On the other hand, for a bad node $t \in V(T)$, $\ell \ge |R_t| \ge \epsilon/3 \cdot |S_t| \ge \epsilon/3 \cdot \OPT(G[H_t])$. Let $U_g$ and $U_b$ denote the set of good and bad nodes, respectively.  Let $S_1 = \bigcup_{t \in U_g} S_t$, and $S_2 = \bigcup_{t \in U_g} R_t$. 
	
	Furthermore, let $V_g = \bigcup_{t \in U_g} \chi(t)$, and $V_b = V \setminus V_g$.  Now, let $S_b \subseteq V_b$ denote an $(1+\epsilon)$-approximate solution to $\Pi$ on the graph $G[V_b]$. We first state the following lemma, whose proof is given after the current proof.
	
	\begin{restatable}{lemma}{deltofapprox} \label{lem:fii-approx}
		$|S| \le (1+\epsilon) \cdot \OPT(G)$, where $S = S_1 \cup S_2 \cup S_b$ as defined above. 
	\end{restatable}
	
	Let $F = G[V_b]$.  From the $\cH$-tree decomposition $(T, \chi, L)$ of $G$ of width $\ell$, it is possible to obtain an $\cH$-tree decomposition of $F$ of width at most $\ell$, by deleting the vertices of $V_g$ from every bag. For simplicity, we will continue to use $(T, \chi, L)$ for the new ``projected'' $\cH$-tree decomposition of $F$. As suggested by \Cref{lem:fii-approx}, we need to find a set $S_b$, which is a $(1+\epsilon)$-approximation to $\Pi$ on $F \coloneqq G[V_b]$. In fact, we will show how to find an optimal solution. To this end, we prove the following technical lemma.
	
	\begin{restatable}{lemma}{LemmaFII} \label{lem:fii-replacement}
		Let $\cH, \cF$ be well-behaved families of graphs, such that $\Pi = \deltof$ is CMSO-definable, and has \FII. Furthermore, suppose there exists an \FPTAS $\pmd$ for $\Pi$, parameterized by $\mdh(\cdot)$ and $\epsilon$. Then, for every positive integer $\ell$, there exists an algorithm $\cA_\ell$, that takes input $(G, T, \chi, L)$, where
		\begin{enumerate}
			\item $G$ is a graph on at most $n$ vertices, and
			\item A $\cH$-tree decomposition $(T, \chi, L)$ of $G$ of width $\ell$, such that:
			\\\qquad $\mathbf{\star}$ for every leaf $t \in V(T)$, the graph $H_t \coloneqq G[\chi(t) \cap L]$ satisfies $\OPT_\Pi(H_t) \le 3\ell/\epsilon$. 
		\end{enumerate}
		The algorithm $\cA_\ell$ runs in time $g(\ell, \epsilon) \cdot \polyn$, where $g$ is some function, and returns an optimal solution $S \subseteq V(G)$ to $\Pi$ on $G$, i.e., $|S| = \OPT_\Pi(G)$, and $G - S \in \cF$.
		%Then, there is an algorithm that, given an $n$-vertex graph $G$, and an integer $k$, decides whether $(G, k) \in \Pi$ in time $f(\twh(G)) \cdot n^{\Oh(1)}$. That is, $\Pi$ is FPT parameterized by $\twh(G)$. 
	\end{restatable}
	The idea of the proof of \Cref{lem:fii-replacement} is similar to that of Theorem 6.1 from \cite{AgrawalKLPRSZ22Elimination}. However, they require that $\pmd$ be an \emph{exact algorithm} for $\Pi$. Our main contribution is to observe that, when the input graph satisfies the property $\star$, an \emph{approximation algorithm} can be made to behave like an exact algorithm by a suitable choice of $\epsilon$. However, a formal proof requires several technical definitions pertaining to \FII, and thus, we defer the formal proof to \Cref{subsec:proof} wherein we first give all the necessary background.  
	
	Now we finish the proof of \Cref{thm:monotone-fii-fptas}. Note that the graph $G[V_b]$ and the $\cH$-tree decomposition $(T, \chi, L)$ restricted to $G[V_b]$ satisfies the requirements of \Cref{lem:fii-replacement} by construction. Thus, using the algorithm from \Cref{lem:fii-replacement}, we can, in fact, get an optimal solution $S_b \subseteq V_b$ to $\Pi$ on $G[V_b]$. This concludes the proof of \Cref{thm:monotone-fii-fptas}, modulo the proof of \Cref{lem:fii-approx}.
	
\end{proof}

\begin{proof}[Proof of \Cref{lem:fii-approx}]
	Recall the definitions of the partial solutions $S_1, S_2, S_b$ as defined in the proof of \Cref{thm:monotone-fii-fptas}. From \Cref{obs:monotone}, we have the following inequality.
	\begin{equation}
		\OPT(G - S_2) \le \OPT(G) \label{eqn:s2-bound}
	\end{equation}
	
	Note that the properties of $\cH$-tree decomposition imply that, if we delete $S_2$ from $G$, the graph $G - S_2$ consists of disjoint induced subgraphs $H_t$ for $t \in U_g$, as well as $G[V_b]$ (note that each of these induced subgraphs may or may not be connected). 
%	 \todo{Reviewer said this line is unclear. I am not sure why.}. Now, consider
	\begin{align*}
		|S| &= |S_1| + |S_2| + |S_b|  \tag{Since $S = S_1 \uplus S_2 \uplus S_b$.}
		\\&\le \sum_{t \in U_g} (|S_t| + |R_t|) + |S_b| \tag{Since $S_1 = \bigcup_{t \in U_g} S_t$ and $S_2 = \bigcup_{t \in U_g} R_t$}
		\\&\le \lr{\sum_{t \in U_g} (1+\epsilon/2) \cdot \OPT(G[H_t])  + \epsilon/2 \cdot \OPT(G[H_t])} + (1+\epsilon) \cdot \OPT(G[V_b]) \tag{Using bounds on $S_t, R_t$ and $S_b$ respectively}
		\\&\le (1+\epsilon) \cdot \lr{ \lr{\sum_{t \in U_g} \OPT(G[H_t])} + \OPT(G[V_b])}
		\\&\le (1+\epsilon) \cdot \OPT(G) \tag{since $\displaystyle V(G)\setminus S_2 = V_b \uplus \biguplus_{t \in U_g} H_t$, and (\ref{eqn:s2-bound})}
	\end{align*}
\end{proof}

%Using Corollary~\ref{cor:planarRdelmod} and Theorem~\ref{thm:monotone-fii-fptas}, we get the following  result. 
%
%\begin{corollary}
%There exists an \FPTAS for \textup{\textsc{Planar $\mathcal{R}$-Deletion}}, that runs in time $f(\ell, \epsilon) \cdot \polyn$, where $\ell = \twh(G)$ (resp. $\ell = \edh(G)$), and $\cH$ is a family of $H$-minor free graphs, for some fixed graph $H$. In particular, this implies an \FPTAS for \textup{\textsc{Treewidth-$\eta$-Modulator}}.
%\end{corollary}\todo{do we need}
%\todo[inline]{Corollaries}

\subsection{Variations of the Equivalence Theorem} \label{subsec:variations}
Theorem~\ref{thm:monotone-fii-fptas} gives  non-uniform \FPTASes. In the following, we obtain a simpler, and \emph{explicit} version of \Cref{thm:monotone-fii-fptas} for \deltof when $\cF$ satisfies the following property, called \emph{treewidth-$\eta$-modulated}: Suppose there exists a constant $\eta$, such that the (standard) treewidth of every graph in $\cF$ is bounded by $\eta$. In this case, we say that the family $\cF$ is \emph{treewidth-$\eta$-modulated}, or simply, \emph{$\eta$-modulated}. Note that many natural \deltof problems are $\eta$-modulated. For example, \textsc{Vertex Cover} (resp.\ \textsc{Feedback Vertex Set}) is $\eta$-modulated for $\eta = 0$ (resp. $\eta = 1$). More generally, it is also known if the family $\cF$ is characterized by a family of excluded minors $\mathcal{O}$ containing at least one planar graph, then \deltof is known as \textsc{Planar $\mathcal{O}$-Deletion}, and is also $\eta$-modulated for some $\eta$ that depends on $\mathcal{O}$.

\begin{restatable}{theorem}{etamodulated} \label{thm:eta-modulated}
	Let $\cH, \cF$ be well-behaved families of graphs, and let $\Pi = \deltof$ such that: (i) $\cF$ is $\eta$-modulated for some constant $\eta \ge 0$, and (ii) there exists an exact \FPT algorithm for $\Pi$ parameterized by $\tw(\cdot)$. Then, the following statements are equivalent.
	\begin{enumerate}
		\item $\Pi$ admits an \FPTAS parameterized by $\mdh(\cdot)$ and $\epsilon$.
		\item $\Pi$ admits an \FPTAS parameterized by $\edh(\cdot)$ and $\epsilon$.
		\item $\Pi$ admits an \FPTAS parameterized by $\twh(\cdot)$ and $\epsilon$.
	\end{enumerate}
\end{restatable}
\begin{proof}
	We proceed exactly as in the proof of \Cref{thm:monotone-fii-fptas}, and classify each node $t \in V(T)$ in the $\cH$-tree decomposition $(T, \chi, L)$ of $G$ into \emph{good} and \emph{bad}, respectively. We use the same notation as in the previous proof.
	
	For every good node $t \in V(T)$, we add an $(1+\epsilon/2)$-approximate solution $S_t$ as well as $R_t$ in the solution. Now we discuss how to find an optimal solution $S_b$ for the graph $G[V_b]$. For each bad node $t \in V(T)$, note that $|R_t| > \epsilon/3 \cdot |S_t|$, which implies that $|S_t| \le 3\ell/\epsilon$. Since $S_t$ is a solution to $\Pi$ on $G[H_t]$, $G[H_t] - S_t \in \cF$, which implies that $\tw(G[H_t] - S_t) \le \eta$. Therefore, $\tw(G[H_t]) \le \eta+|S_t| \le 3\ell/\epsilon + \eta$. This implies that $\tw(G[V_b]) \le \ell+ 3\ell/\epsilon + \eta$. Furthermore, the given $\cH$-tree decomposition $(T, \chi, L)$ of $G[V_b]$ can be transformed into the (standard) tree decomposition of width at most $\ell + 3\ell/\epsilon + \eta$. Then, by using the \FPT algorithm for $\Pi$ parameterized by $\tw(G[V_b])$, the theorem follows.
\end{proof}

Note that \Cref{thm:eta-modulated} requires that $\Pi$ admit an \FPTAS parameterized by $\mdh(G)$. In particular, when $\mdh(G) = 0 \iff G \in \cH$, $\Pi$ should have an \EPTAS. However, in some cases, we only know an $\alpha$-approximation for $\Pi$ even when $G \in \cH$, where $\alpha \ge 1$. In such cases, the following theorem may be applicable, if $\Pi$ is also known to be \FPT parameterized by the size of the solution.

\begin{theorem} \label{thm:constant-approx-equiv}
	Let $\cH, \cF$ be well-behaved families of graphs, such that $\Pi = \deltof$ satisfies the following properties: (i) $\Pi$ is CMSO-definable, (ii) has \FII, and (iii) $\Pi$ is \FPT parameterized by the size of the solution. Also, suppose that for some constant $\alpha \ge 1$, $\Pi$ admits an $\alpha$-approximation in time $f(\mdh(\cdot)) \cdot \polyn$. Then, 
	\begin{enumerate}
		\item $\Pi$ admits an $(\alpha+\epsilon)$-approximation in time $g(\edh(\cdot), \epsilon) \cdot \polyn$, and
		\item $\Pi$ admits an $(\alpha+\epsilon)$-approximation in time $h(\twh(\cdot), \epsilon) \cdot \polyn$.
	\end{enumerate}
\end{theorem}
\begin{proof}[Proof Sketch]
	The proof of the theorem follows essentially follows the same structure as that of \Cref{thm:monotone-fii-fptas}, line-by-line. We only highlight the important differences. First, we use the $\alpha$-approximation on each $G[H_t]$ to obtain a set $S_t \subseteq H_t$ such that $|S_t| \le \alpha \cdot \OPT(G[H_t])$, for each node $t \in V(T)$ in the $\cH$-tree decomposition $(T, \chi, L)$. Then, a node is said to be \emph{good} if $|S_t| \le (\epsilon/2\alpha) \cdot |S_t| \le \epsilon/2 \cdot \OPT(G[H_t])$. Otherwise, if $|S_t| > \epsilon/(2\alpha) \cdot \OPT(G[H_t])$, then $t$ is bad. 
	
	We follow the same notation as in the proof of \Cref{thm:monotone-fii-fptas}, and let $S_1 = \bigcup_{t \in U_g} S_t, S_2= \bigcup_{t \in U_g} R_t$, and $S_b$ be an $(1+\epsilon)$-approximation to $\Pi$ on $G[V_b]$. Then, we obtain an analogous version of \Cref{lem:fii-approx} implies that $|S| \le (\alpha +\epsilon) \cdot \OPT(G)$, where $S = S_1 \cup S_2 \cup S_b$. Thus, it suffices to obtain an $(1+\epsilon)$-approximation to $\Pi$ on $S_b$.
	
	Now, we can use the FPT algorithm for $\Pi$ parameterized by the solution size to prove an analogous version of \Cref{lem:fii-replacement} (this, in fact, almost exactly matches Theorem 6.1 from \cite{AgrawalKLPRSZ22Elimination}). Here, we use the fact that, when we perform a replacement (i.e., application of \Cref{lem:red2finiteindex}) using the \FII property of $\Pi$, the size of the solution is upper bounded by $\OPT(G[H_t]) + c_{\mdh(G)}$, which is upper bounded by a function of $\twh(G)$ and $\epsilon$, where $t$ is a \emph{bad node}. Thus, in time $h'(\twh, \epsilon) \cdot \polyn$, we obtain an \emph{exact solution} $S_b$ to $\Pi$ on $G[V_b]$. Thus, overall, we obtain an $(\alpha+\epsilon)$-approximation to $\Pi$ on $G$. 
\end{proof}

Plugging in \Cref{cor:oct} into \Cref{thm:constant-approx-equiv}, we obtain the following corollary.

\begin{corollary}
	There exists an algorithm that runs in time $f(\ell,\epsilon) \cdot \polyn$, where $\ell$ is $\edh(G)$ (or $\twh(G)$) to planar graphs, and computes a $(\frac{9}{4} + \epsilon)$-approximation for \textup{\textsc{Odd Cycle Transversal}}.
\end{corollary}

%\begin{corollary}
%\todo[inline]{Corollary for OCT on planar graphs}

%% file: packing.tex
%!TEX root = main.tex
\newcommand{\size}{\mathsf{size}}
\newcommand{\cc}{\texttt{cc}}
\section{Packing Problems} \label{sec:packing}

In this section, we consider packing problems. First, we give preliminaries on packing problems in \Cref{subsec:packing-prelim}. Then, in \Cref{subsec:packing-fptas}, we design \FPTASes for these packing problems, parameterized by $\mdh(\cdot)$. Finally, in \Cref{subsec:packing-equiv}, we prove the equivalence theorems for these packing problems, which extends these \FPTASes, parameterized by $\edh(\cdot)$, and $\twh(\cdot)$.

\subsection{Preliminaries on Packing Problems} \label{subsec:packing-prelim}
We consider the following three packing problems in this section.
\begin{tcolorbox}[colback=white!5!white,colframe=gray!75!black]
	\indset
	\\\textbf{Input:} An instance $(G, k)$, where $G = (V, E)$ is a graph, and $k$ is a positive integer.
	\\\textbf{Question:} Does there exist a subset $S \subseteq V(G)$ of size at least $k$  that is independent in $G$?
\end{tcolorbox}

\begin{tcolorbox}[colback=white!5!white,colframe=gray!75!black]
	\fpacking ($\cF$ is a fixed finite family of graphs)
	\\\textbf{Input:} An instance $(G, k)$, where $G = (V, E)$ is a graph, and $k$ is a positive integer.
	\\\textbf{Question:} Do there exist $k$ pairwise vertex-disjoint subgraphs $H_1, \ldots, H_k$ of $G$, such that each $H_i$ contains a graph from $\cF$ as a minor?
\end{tcolorbox}

\begin{tcolorbox}[colback=white!5!white,colframe=gray!75!black]
	\spacking ($\cS$ is a fixed finite family of graphs)
	\\\textbf{Input:} An instance $(G, k)$, where $G = (V, E)$ is a graph, and $k$ is a positive integer.
	\\\textbf{Question:} Do there exist $k$ pairwise vertex-disjoint subgraphs $H_1, \ldots, H_k$ of $G$, such that each $H_i$ is isomorphic to some graph in $\cS$?
\end{tcolorbox}
Throughout the section, we will assume the families $\cF$ and $\cS$ in \fpacking and \spacking, are fixed finite families of connected graphs. First, we note that \indset is self-reducible. This can be inferred from, e.g., the self-reduicibility of \vc (\Cref{lem:self-reducibility}). Now we prove the same for \fpacking and \spacking, under the mild assumption that the graphs in $\cF$ and $\cS$ are connected and non-empty. For this, we need some definitions.

\begin{definition}[Minor Model] \label{def:minor-model}
	$H$ is a minor of $G$ iff there exists a function $\varphi: V(G) \to 2^{V(G)} \setminus \{\emptyset\}$ such that for every $u \in V(H)$, (i) $G[\varphi(u)]$ is connected, and (ii) for every $uv \in E(H)$, there exist $u' \in \varphi(u)$ and $v' \in \varphi(v)$, such that $u'v' \in E(G)$. Such a function $\varphi$ is called a \emph{minor model} of $H$ in $G$.
\end{definition}

\begin{definition}[Graph Isomorphism] \label{def:isomorphism}
	The graphs $H$ and $G$ are said to be \emph{isomorphic} if there exists a bijection $\varphi: V(H) \to V(G)$ such that for any $u, v \in V(H)$, $uv \in E(H)$ iff $\varphi(u) \varphi(v) \in E(G)$.
\end{definition}

\begin{definition}[Solution to \fpacking]
	Let $\cF$ be a fixed finite family of connected graphs. Let $(G, k)$ be a yes-instance of \fpacking. Then, a solution to $(G, k)$ is given by $\sigma = \LR{(H_1, F_1, \varphi_1), (H_2, F_2, \varphi_2), \ldots, (H_k, F_k, \varphi_k)}$, with the following properties. 
	\begin{itemize}
		\item $H_1, \ldots, H_k$ are pairwise vertex-disjoint subgraphs of $G$.
		\item $F_1, \ldots, F_k$ are (not necessarily distinct) graphs from $\cF$.
		\item For each $i \in [k]$, $\varphi_i$ is a minor model of $F_i$ in $H_i$.
	\end{itemize} 
	We refer that each tuple $(H_i, F_i, \varphi_i)$ as a \emph{minor model tuple}. The \emph{size} of the solution $\sigma$, is defined as the number of minor model tuples in the solution, e.g., $\size(\sigma) \coloneqq k$ for $\sigma$ above. For a graph $G$, let $\OPT_{\Pi}(G)$ denote the maximum size of a solution to $\Pi \coloneqq \fpacking$. An $\alpha$-approximation algorithm returns a solution of size at least $\alpha \cdot \OPT(G)$ on any input graph $G$. 
	These notions are analogously defined for \spacking, where we require $\varphi_i$'s to be isomorphisms instead of minor models.
\end{definition}

\begin{lemma} \label{lem:self-reducibility-packing}
	For any fixed finite family of connected graphs $\cF$ (resp. $\cS$), \fpacking (resp.\ \spacking) is self-reducible. That is, given an algorithm for the decision version, there exists an algorithm that produces a solution to an instance $(G, k)$, or concludes that $(G, k)$ is a no-instance.
\end{lemma}
\begin{proof}
	We prove the theorem for \fpacking and then discuss the case of \spacking, which is only simpler. Let $\cA$ be an algorithm for the decision version of \fpacking, which we denote by $\Pi$ for brevity. 
	
	Consider an input $(G, k)$. First, we run $\cA$ on $(G, k)$, and if it outputs that it is a no-instance, then we report the same. Otherwise, suppose that $(G, k)$ is a yes-instance. Then, there exists a solution $\LR{(H_i, F_i, \varphi_i)}_{i \in [k]}$. Let $U \coloneqq \bigcup_{i \in [k]} V(H_i)$, and let $F = \bigcup_{i \in [k]} E(H_i)$. We prove the following claim, which follows quite easily from the definitions of a minor model and of a solution.
	\begin{observation} \label{cl:vertex-edge-deletion}
		Suppose that for some $u \in V(G)$, $(G - u, k)$ is a yes-instance. Then, any solution for $(G-u, k)$ for $\Pi$ is a solution for $(G, k)$. Similarly, suppose that for some $uv \in V(G)$, $(G - uv, k)$ is a yes-instance. Then, any solution for $(G-uv, k)$ for $\Pi$ is a solution for $(G, k)$. 
	\end{observation}
	Let $G'$ be a copy of $G$. We first iterate over each vertex $u \in V(G')$, and check whether $(G - u, k)$ is a yes-instance using $\cA$. If it is a yes-instance, then we redefine $G' \gets G' - u$, and proceed with the next iteration; otherwise we leave the graph unchanged. Next, we do the same for each edge $uv \in E(G')$, i.e., check whether $(G' - uv, k)$ is a yes-instance using $\cA$, and delete the edge from $G'$ is that is the case. Note that we make at most $|V(G)| + |E(G)|$ calls to $\cA$. If $G'$ denotes the graph at the end of this procedure, then \Cref{cl:vertex-edge-deletion} implies that $(G', k)$ is a yes-instance of \fpacking.

	Now, consider a solution $\sigma = \{(G'_i, F_i, \varphi_i)\}_{i \in [k]}$ of $\Pi$ on $G'$. First, we have the following claim.

	\begin{claim}
		For every $(G'_i, F_i, \varphi_i) \in \sigma$, there exists a unique connected component $C$ of $G'$ such that $V(G'_i) = C$.
	\end{claim}
	\begin{proof}
		First we show that the number of connected components in $G'$ is exactly $k = |\sigma|$. Suppose $|\cc(G')| > k$, then there exists a non-empty connected component $C \in \cc(G)$ such that $C \cap V(G'_i) = \emptyset$ for all $i \in [k]$. Then, for any $v \in C$, $\sigma$ remains a solution for $(G - v, k)$, which contradicts the assumption on $G'$. Now, suppose $|\cc(G')| < k$. Then, by pigeonhole principle, there exists a connected component $C$ such that $C \cap V(G'_i) \neq \emptyset$ and $C \cap V(G'_j) \neq \emptyset$ for two distinct $G'_i, G'_j$. Thus, there exists a path between vertices $u$ and $v$ in $(G'[C])$ such that $u \in V(G'_i)$ and $v \in V(G'_j)$. Now, there are two cases. If there exists an edge $u'v' \in E(G'[C])$ such that $u' \in V(G'_i)$ and $v' \in V(G'_j)$, then we observe that $\sigma$ is a solution for $(G - uv, k)$, which is a contradiction. Otherwise, there is an edge $u'v' \in E(G'[C])$ such that $u' \in V(G'_i)$, but $v' \not\in \bigcup_{j' \in [k]} V(G'_{j'})$. This is again a contradiction since $\sigma$ is a solution $(G-v', k)$. Thus, the number of connected components of $G'$ is exactly $k'$, and each connected component contains exactly the vertices from $V(G'_i)$ for some $i \in [k]$. 
	\end{proof}
	Thus, consider a connected subgraph $G'_i$ of $G'$ corresponding to a minor model tuple $(G'_i, F_i, \varphi_i)$, and an edge $uv \in E(G'_i)$. If $\varphi(u) = \varphi(v)$, then $(G'\backslash uv, k)$ is a yes-instance. Similarly, if for some $u'v' \in E(G'_i)$, $(G' \backslash uv, k)$ is a yes-instance, then the solution for $(G' \backslash uv, k)$ can be translated into a solution for $(G', k)$. Therefore, we iterate over the edges of $G'_i$, and use $\cA$ to decide the instance $(G' \backslash uv, k)$. If it returns yes, we contract the edge and proceed to the next iteration. We stop when none of the edges can be contracted. It follows that the resulting graph is isomorphic to some $F_i \in \cF$. By iterating over all connected components, we obtain a solution of size $k$, which can be translated into a solution for $G$.
\end{proof}

\input{packing-examples}

\input{packing-maintheorem}

%% file: packing-examples.tex
%!TEX root = main.tex

\subsection{FPT-ASes for Packing Problems Parameterized by $\mdh$} \label{subsec:packing-fptas}

Here, we show that a natural variant of the bucket versus ocean also yields \FPTAS{}es for packing problems such as \textsc{Independent Set} and \textsc{Cycle Packing}.  Let $\cH$ be an apex-minor closed family. We are given a graph $G = (V, E)$, and a set $M \subseteq V(G)$ of size at most $p$ such that $G' = G\setminus M \in \cH$, Let $H = V(G) \setminus M$.

\medskip\noindent\textbf{\textsf{Independent Set.}} By iterating over all subsets of $M$, we guess the intersection of an optimal independent set $I$ with $M$. Consider the iteration corresponding to $Y = I \cap M$. Then, let $R = N(Y) \cap H$. We compute a $(1-\epsilon)$-approximation $I'$ to the maximum independent set in the graph $G' \setminus R$, using an $2^{\Oh(1/\epsilon)} \cdot \polyn$ time algorithm from Fomin et al.~\cite{FominLS18Grid}. It follows that $|Y| + |I'| \ge |Y| + (1-\epsilon) \cdot |I \setminus Y| \ge (1-\epsilon) \cdot |I| = \OPT(G)$. Thus, we get the following result.

\begin{theorem} \label{thm:indset-fptas-modh}
	Let $\cH$ be an apex-minor free graph family. Then, there exists an \FPTAS for \indset parameterized by $\mdh(\cdot)$ that runs in time $2^{\Oh(\mdh(G) + 1/\epsilon)} \cdot \polyn$.
\end{theorem}
%\medskip
\noindent\textbf{\textsf{Cycle Packing.}} We compute an $(1-\epsilon/2)$-approximate solution $S$ to cycle packing on the graph $G'$ in time $f(\epsilon) \cdot \polyn$. Note that \textsc{Cycle Packing} is a special case of \fpacking, where $\cF$ contains a single graph $C_3$ (i.e., a cycle on $3$ vertices).  We begin by stating the following two inequalities.
\begin{align}
	\label{eqn:cyclepacking1} \OPT(G)   \ge  \OPT(G') \ge |S| \ge (1-\frac{\epsilon}{2}) \cdot \OPT(G') \\
%\end{equation}
%\begin{equation}
	\label{eqn:cyclepacking2}\OPT(G') + |M|   \ge  \OPT(G) \ge \OPT(G') 
\end{align}
This follows from the fact that any solution to \textsc{Cycle Packing} on $G'$ is also a solution to $G$. On the other hand, $G' = G \setminus X$. Now consider a set $S^*$ of disjoint cycles in $G$. Note that at most $|M|$ cycles are incident to a vertex of $X$. Therefore, it is possible to obtain a set $\tilde{S}^*$ of disjoint cycles in $G'$ of size at least $|S^*| - |M|$. Now, consider two cases.

\textbf{Case 1.} If $|M| \le \frac{\epsilon}{2} \cdot |S|$. We have the following lemma.
\begin{claim}
	If $|M| \le \frac{\epsilon}{2} \cdot |S|$, then $|S| \ge (1-\epsilon) \OPT(G)$.
\end{claim}
\begin{proof}
	Now, consider:
	\begin{align*}
		|S| \ge (1- \frac{\epsilon}{2}) \OPT(G') &\ge (1-\frac{\epsilon}{2}) \cdot \lr{\OPT(G) - |M|} \tag{From (\ref{eqn:cyclepacking2})}
		\\&\ge (1-\frac{\epsilon}{2}) \cdot \OPT(G) - (1-\frac{\epsilon}{2}) \cdot \frac{\epsilon}{2} |S| \tag{Since $|M| \le \frac{\epsilon}{2}|S|$}
		\\&\ge (1-\frac{\epsilon}{2}) \cdot \OPT(G) - \frac{\epsilon}{2} \cdot \OPT(G) \tag{From (\ref{eqn:cyclepacking1})}
		\\&= (1-\epsilon) \cdot \OPT(G)
	\end{align*}
\end{proof}

\textbf{Case 2.} $|M| \ge \frac{\epsilon}{2}|S|$. Then, 
\begin{lemma}
	If $|M| \ge \frac{\epsilon}{2} |S|$, then $\OPT(G) \le \lr{1 + \frac{4}{\epsilon}} p$.
\end{lemma}
\begin{proof}
	Note that $|M| \ge \frac{\epsilon}{2} \cdot |S| \ge (1-\frac{\epsilon}{2}) \cdot \frac{\epsilon}{2} \cdot \OPT(G') \ge \frac{\epsilon}{4} \OPT(G')$, where the second-last inequality follows from (\ref{eqn:cyclepacking1}). In other words, $\OPT(G') \le \frac{4}{\epsilon} |M|$. Now, consider
	\begin{align*}
		\OPT(G) \le \OPT(G') + |M| \le \lr{1 + \frac{4}{\epsilon}} |M| = \lr{1 + \frac{4}{\epsilon}} p.
	\end{align*}
\end{proof}
In this case, we can use an algorithm for \textsc{Cycle Packing} parameterized by the solution size that runs in time $2^{\Oh(\frac{k \log^2 k}{\log\log k})} \cdot \polyn$ time~\cite{LokshtanovMSZ19}. Therefore, we obtain the following result.

\begin{corollary}
	Let $\cH$ be a family of apex-minor of graphs. Then, there exists an \FPTAS for \textup{\textsc{Cycle Packing}} that runs in time $2^{\tilde{\Oh}({\mdh(G)/\epsilon})} \cdot \polyn$.
\end{corollary}

%% file: packing-maintheorem.tex
%!TEX root = main.tex
\subsection{Equivalence Theorems for Packing Problems} \label{subsec:packing-equiv}
First we prove an equivalence theorem for \indset (\Cref{thm:indset-fptas}), and then proceed to proving a similar theorem for \fpacking and \spacking (\Cref{thm:fpacking-fptas}), whose proof is technically more involved.

\subsubsection{Independent Set}

\begin{theorem} \label{thm:indset-fptas}
	Let $\cH$ be a well-behaved family of graphs, and let $\Pi = \indset$. Then, the following statements are equivalent.
	\begin{enumerate}
		\item $\Pi$ admits an \FPTAS parameterized by $\mdh(\cdot)$ and $\epsilon$.
		\item $\Pi$ admits an \FPTAS parameterized by $\edh(\cdot)$ and $\epsilon$.
		\item $\Pi$ admits an \FPTAS parameterized by $\twh(\cdot)$ and $\epsilon$.
	\end{enumerate}
\end{theorem}
\begin{proof}
	
	The overall structure of the proof follows the same outline as the proof of \Cref{thm:monotone-fii-fptas}, so we only discuss the modifications required. 
	As before, it suffices to prove $1 \implies 3$. %Most of the proof goes through regardless of the choice of $\Pi$, so we discuss the specific $\Pi$ only to discuss the differences. 
	We also omit the subscript $\Pi$ from $\OPT_\Pi(\cdot)$. Let $\pmd$ denote the supposed \FPTAS for $\Pi$. 
	
	Let $G = (V, E)$ be a graph, and let $(T, \chi, L)$ be the given $\cH$-tree decomposition of width $\ell-1$. For each node $t \in V(T)$, let $\chi(t) = H_t \uplus R_t$, where $H_t \subseteq L$ with $G[H_t] \in \cH$, and $|R_t| \le \ell$.  
By iterating over each leaf node $t \in V(T)$, we classify as \emph{good} or \emph{bad}, as follows. We use \pmd on $G[H_t]$ to obtain an independent set $S_t \subseteq H_t$ such that $|S_t| \ge (1-\epsilon/2) \cdot \OPT(G[H_t])$, which takes $f(0, \epsilon/2) \cdot |H_t|^{\Oh(1)} = h(\epsilon) \cdot \polyn$ time. Now, if $|R_t| \le \epsilon/2 \cdot |S_t|$, then we say that the leaf node $t$ is \emph{good}; otherwise we say that $t$ is \emph{bad}. Furthermore, all internal nodes are classified as bad. Let $U_g, U_b$ denote the sets of good and bad nodes respectively, and let $S_1 = \bigcup_{t \in U_g} S_t$. Finally, let $V_g = \bigcup_{t \in U_g} \chi(t)$ and $V_b = V \setminus V_g$. Let $S_b \subseteq V_b$ denote an $(1-\epsilon)$-approximate solution to $\Pi$ on $G[V_b]$. We prove the following lemma.
	
	\begin{lemma} \label{lem:indset-appr}
		$|S| \ge (1-\epsilon) \cdot \OPT(G)$, where $S \coloneqq S_1 \cup S_b$, as defined above. 
	\end{lemma}
	\begin{proof}
		Let $S^* \subseteq V(G)$ denote a maximum-size independent set. Let $S^*_1 \coloneqq S^* \cap \bigcup_{t \in U_g} H_t$, $S^*_2 \coloneqq S^* \cap \bigcup_{t \in U_g} R_t$, and $S^*_b \coloneqq S^* \cap V_b$. Note that $\{S^*_1, S^*_2, S^*_3\}$ is a partition of $S^*$. 
		
		For each $t \in U_g$, let $S^*_{1, t} \coloneqq S^*_1 \cap H_t$, and $S^*_{2, t} \coloneqq S^*_2 \cap R_t$. Note that $\LR{S^*_{1, t}}_{t \in U_g}$ is a partition of $S^*_1$, whereas a vertex $u$ may belong to multiple $S^*_{2, t}$'s. Also, note that for any $t\in U_g$, $S^*_{1, t}$ is an independent set in $G[H_t]$, which implies that $|S^*_{1, t}| \le \OPT(G[H_t])$. Now, consider
		\begin{align}
			|S^*_1| + |S^*_2| &\le \sum_{t \in U_g} |S^*_{1, t}| + |S^*_{2, t}| \nonumber
			\\&\le \sum_{t \in U_g} \OPT(G[H_t]) + |R_t| \tag{Since $S^*_{2, t} \subseteq R_t$ }
			\\&\le \sum_{t \in U_g} \frac{1}{1-\epsilon/2} \cdot |S_t| + \frac{\epsilon}{2} \cdot |S_t| \tag{Since $S_t$ is $1-\epsilon/2$-approximation, and by definition of good node}
			\\&\le (1+\epsilon) \cdot \sum_{t \in U_g} |S_t| \nonumber
			\\&\le \frac{1}{1-\epsilon} \cdot |S_1|  \label{eqn:indset-1}
		\end{align}
		Note that $S^*_b$ is an independent set in $G[V_b]$, which implies that 
		\begin{align}
			|S^*_b| \le \OPT(G[V_b]) \le \frac{1}{1-\epsilon} \cdot |S_b| \label{eqn:indset-2}
		\end{align}
		Therefore, 
		\begin{align}
			\OPT(G) = |S^*| &= |S^*_1| + |S^*_2| + |S^*_b| \nonumber
			\\&\le \frac{1}{1-\epsilon} \cdot |S_1| + \frac{1}{1-\epsilon} \cdot |S_b| \tag{from (\ref{eqn:indset-1}) and (\ref{eqn:indset-2})}
			\\&\le \frac{1}{1-\epsilon} \cdot |S|
		\end{align}
		Which concludes the proof.
	\end{proof}
	Now, since \indset has \FII, and the size of an optimal solution in each $G[H_t]$ for $t \in U_b$ is bounded, we can prove a lemma analogous to \Cref{lem:fii-replacement}, which uses \pmd to return an optimal solution $S_b \subseteq G[V_b]$ in time $g(\twh(G), \epsilon) \cdot \polyn$. This concludes the proof of the theorem.
\end{proof}

\subsubsection{\fpacking and \spacking}

Now we adapt the previous proof to \fpacking and \spacking.
\begin{theorem} \label{thm:fpacking-fptas}
	Let $\cH$ be a well-behaved family of graphs. Let $\Pi = \fpacking$ (resp.\ \spacking), where $\cF$ ($\cS$) is a fixed finite family of connected graphs, such that $\Pi$ is CMSO-definable and has \FII. Then, the following statements are equivalent.
	\begin{enumerate}
		\item $\Pi$ admits an \FPTAS parameterized by $\mdh(\cdot)$ and $\epsilon$.
		\item $\Pi$ admits an \FPTAS parameterized by $\edh(\cdot)$ and $\epsilon$.
		\item $\Pi$ admits an \FPTAS parameterized by $\twh(\cdot)$ and $\epsilon$.
	\end{enumerate}
\end{theorem}
\begin{proof}
	We focus on $\Pi \coloneqq \fpacking$, since the arguments for \spacking are only easier. Also, as before, we prove $1 \implies 3$. Let $\pmd$ be the \FPTAS for $\Pi$ parameterized by $\mdh(\cdot)$ and $\epsilon$. 
	
	Let $G = (V, E)$ be the input graph on $n$ vertices, and $(T, \chi, L)$ be its $\cH$-tree decomposition of width $\ell$. By iterating over each leaf node $t \in V(T)$, we use \pmd on $G[H_t]$, to obtain a solution $\sigma_t \coloneqq \LR{(G_{i, t}, F_{i, t}, \varphi_{i, t})}_{i \in [k_t]}$ for some $k_t \ge 0$, such that $|\sigma_t| \ge (1-\epsilon) \cdot \OPT(G[H_t])$. Let $S_t \coloneqq \bigcup_{i \in [k_t]} V(G_{i, t})$. If $|R_t| \le \epsilon/2 \cdot |\sigma_t|$, then we say that the leaf node $t$ is \emph{good}; \emph{bad} otherwise. All internal nodes are classified as \emph{bad}. As before, let $U_g, U_b$ denote the sets of good and bad nodes respectively. These definitions follow the same idea as in the proof of \Cref{thm:indset-fptas} so far. Now, we slightly deviate, and use $V_g \coloneqq \bigcup_{t \in U_g} H_t$, and let $V_b = V(G) \setminus V_g$. (Note that these definitions are different from \Cref{thm:indset-fptas}, where we used $\bigcup_{t \in U_g} \chi(t)$. That is, the vertices of $R_t$ belong to $V_b$ instead of $V_g$.)
	Let $\sigma_1 = \bigcup_{t \in U_g} \sigma_t$, and let $\sigma_b$ denote an $(1-\epsilon)$-approximate solution to $\Pi$ on $G[V_b]$. We claim the following
	
	\begin{lemma} \label{lem:fpacking-appr}
		$|\sigma| \ge (1-\epsilon) \cdot \OPT(G)$, $\sigma = \sigma_1 \cup \sigma_b$, as defined above.
	\end{lemma}
	\begin{proof}
		Let $\sigma^*$ denote an optimal solution to $\Pi$ on $G$. Let us classify the minor models from $\sigma^*$ into three types. Consider some minor model tuple $M = (G_i, F_i, \varphi_i) \in \sigma^*$ that is a minor model tuple. Note that the subgraphs $G_i$ are pairwise vertex-disjoint. Then, 
		\begin{itemize}
			\item $M$ is \textbf{type 1} if for some good node $t \in U_g$, $V(G_i) \subseteq H_t$.
			\item $M$ is \textbf{type 2} if for some good node $t \in U_g$, it holds that $V(G_i) \cap H_t \neq \emptyset$, and $V(G_i) \setminus H_t \neq \emptyset$.
			\\%Note that 
			Since $\cF$ contains connected graphs, this implies that, in particular, $V(G_i) \cap R_t \neq \emptyset$.
			\item $M$ is \textbf{type 3} if $V(G_i) \subseteq V_b$.
		\end{itemize}
		Let $\sigma^*_1, \sigma^*_2, \sigma^*_3$ denote the subsets of $\sigma$ containing minor model tuples of \textbf{type 1, 2, 3} respectively.  For each $t \in U_g$, let $\sigma^*_{1, t} = \LR{ (G_i, F_i, \varphi_i) \in \sigma^*_1 : V(G_i) \subseteq H_t }$, and note that $\LR{\sigma^*_{1, t}}_{t \in U_g}$ is a partition of $\sigma^*_1$. Note that 
		\begin{equation}
			|\sigma^*_{1, t}| \le \OPT(G[H_t]) \label{eqn:fpacking-1}
		\end{equation}
		
		Similarly, for each $t \in U_g$, let $\sigma^*_{2, t} = \LR{ (G_i, F_i, \varphi_i) \in \sigma^*_2 : V(G_i) \cap H_t \neq \emptyset }$. Note that $\LR{\sigma^*_{2, t}}_{t \in U_g}$ is not necessarily partition of $\sigma^*_2$, since a minor model from $\sigma^*_2$ may span multiple $G[H_t]$'s. We first have the following observation.
		\begin{observation} \label{obs:minor-models-small}
			For any $t \in U_g$, $|\sigma^*_{2, t}| \le |R_t| \le \epsilon/2 \cdot |\sigma_t|$.
		\end{observation}
		\begin{proof}
			For each $(G_i, F_i, \varphi_i) \in \sigma^*_{2, t}$, note that $V(G_i) \cap R_t \neq \emptyset$. However, since $\sigma^*_{2, t} \subseteq \sigma^*$, it follows that the vertex sets of $G_i$'s from $\sigma^*_{2, t}$ are pairwise disjoint. This implies that $|\sigma^*_{2, t}| \le |R_t|$. Finally, since $t$ is a good node, by definition, we have that $|R_t|\le \epsilon/2 \cdot |\sigma_t|$.
		\end{proof}
		Now, we proceed exactly as in the proof of \Cref{lem:indset-appr}. Consider,
		\begin{align}
			|\sigma^*_1| + |\sigma^*_2| &\le \sum_{t \in U_g} |\sigma^*_1, t| + |\sigma^*_{2, t}| \nonumber
			\\&\le \sum_{t \in U_g} \OPT(G[H_t]) + \frac{\epsilon}{2} \cdot |\sigma_t| \tag{From (\ref{eqn:fpacking-1}) and \Cref{obs:minor-models-small}}
			\\&\le \sum_{t \in U_g} \frac{1}{1-\epsilon/2} |\sigma_t| + \frac{\epsilon}{2} |\sigma_t|
			\\&\le \frac{1}{1-\epsilon} \cdot \sum_{t \in U_g} |\sigma_t| = \frac{1}{1-\epsilon} |\sigma_1| \label{eqn:fpacking-2}
		\end{align}
		Since $\sigma^*_b$ is a valid solution to $\Pi$ on $G[V_b]$, it follows that
		\begin{equation}
			|\sigma^*_b| \le \OPT(G[V_b]) \le \frac{1}{1-\epsilon} |\sigma_b| \label{eqn:fpacking-3}
		\end{equation}
		Finally, consider,
		\begin{align*}
			\OPT(G) = |\sigma^*| &= |\sigma^*_1| + |\sigma^*_2| + |\sigma^*_3| \\&\le \frac{1}{1-\epsilon} |\sigma_1| + \frac{1}{1-\epsilon} |\sigma_b| \tag{From (\ref{eqn:fpacking-2}) and (\ref{eqn:fpacking-3})}
			\\&= \frac{1}{1-\epsilon} |\sigma|
		\end{align*}
		This concludes the proof of the lemma.
	\end{proof}
	Now, the task is reduced to finding a solution $\sigma_b$ to $\Pi$ on $G[V_b]$ such that $|\sigma_b| \ge (1-\epsilon) \cdot \OPT(G[V_b])$. As before, we can prove an analogous version of \Cref{lem:fii-replacement}, where we use $\pmd$ and the fact that $\Pi$ is CMSO-definable, to obtain an optimal solution on $G[V_b]$. Again, the crucial observation is that each replacement operation (i.e., application \Cref{lem:red2finiteindex}) requires solving $\Pi$ exactly on an instance, where the instances are of the form $G[\chi(t)] \oplus Y$, where $|V(Y)| \le g(\ell, \epsilon)$. Since $\OPT(G[\chi(t)]) \le \OPT(G[H_t]) + |R_t| \le 2\ell/\epsilon + \ell$, and each vertex in $Y$ can contribute to the image of at most one minor model. Therefore, $\OPT(G[H_t \oplus Y])$ is upper bounded by a function of $\ell$ and $\epsilon$. Therefore, by suitably choosing the value of $\epsilon'$, $\pmd$ can be made to return an exact solution on such instance in time \FPT in $\ell$ and $\epsilon$. This concludes the proof of the theorem for \fpacking. Finally, we note that all arguments only rely on assumption that the graphs in $\cF$ are connected, and thus also hold for \spacking.
\end{proof}

%% file: domset.tex
\section{\FPTASes for Dominating Set} \label{sec:domset}
In this section, we consider \textsc{Dominating Set}, defined as follows. 
\begin{tcolorbox}[colback=white!5!white,colframe=gray!75!black]
	\ads
	\\\textbf{Input:} An instance $(G, k)$, where $G$ is a graph, and $k$ is a non-negative integer
	\\\textbf{Question:} Does $G$ contain a dominating set of size at most $k$, i.e., does there exist a set $S \subseteq V(G)$ such that for each $u \in V(G)$, $N[v] \cap S \neq \emptyset$?
\end{tcolorbox}
Note that the main difficulty is that \textsc{Dominating Set} is not monotone, i.e., it does not necessarily hold that for any $S \subseteq V(G)$, $\OPT(G - S) \le \OPT(G)$. Thus, the approach for \deltof does not immediately generalize, and the arguments for \domset are technically more involved. Furthermore, it is not immediately clear that \domset is \emph{self-reducible}. Therefore, we need to rely on different \emph{annotated versions} of \domset to design approximation algorithms for the problem.

First, in \Cref{subsec:fptas-domset}, we design an \FPTAS for \domset, parameterized by the $\mdh(\cdot)$ and $\epsilon$, for any apex-minor free family $\cH$. Then we prove the equivalence theorem (cf. \Cref{thm:equiv-domset}), which implies \FPTAS for \domset, parameterized by $\edh(\cdot)$ (resp.\ $\twh(\cdot)$) and $\epsilon$. The proof of this theorem is conceptually similar to that of \Cref{thm:monotone-fii-fptas}; however, it is technically much more involved due to the reasons mentioned above. 

\input{domset-examples}
\input{maintheorem-domset}

%% file: domset-examples.tex
%!TEX root = main.tex
\subsection{\FPTAS for \domset Parameterized by $\mdh$} \label{subsec:fptas-domset}

In the following, we combine the case analysis from the previous approach, as well as use arguments involving bidimensionality, as in Fomin et al.\ \cite{FominLS18Grid}, to design an \FPTAS for \textsc{Dominating Set} parameterized by $\mdh(\cdot)$ and $\epsilon$, where $\cH$ is a family of apex-minor free graphs, which we fix in the following discussion. Before proceeding to the algorithm for \domset, we consider an annotated version of \domset and discuss its properties.

\medskip\noindent\textbf{\textsf{Annotated Dominating Set.}} We are given a graph $G'$, and a set of vertices $D \subseteq V(G')$ that are marked as \emph{dominated}, and the vertices of $V(G') \setminus D$ are marked as \emph{undominated}. Furthermore, $G'[D]$ is edgeless. The goal of \anndomset is to compute a minimum-size set of vertices that dominates all the vertices of $V(G') \setminus D$, i.e., to compute a dominating set for all the undominated vertices. We use $(G', D)$ to denote this instance of \anndomset. 

In the definition of \anndomset, we have enforced that $G'[D]$ is edgeless (i.e., the dominated vertices induce an independent set). This assumption is convenient in the subsequent arguments, and is without loss of generality, for the following reason.
%is convenient to assume that the set of dominated vertices $D$ forms an independent set in $G'$, and this assumption is without loss of generality for the following reason. 
Note that a solution for $(G', D)$ is allowed to contain vertices from $D$, but it is not required to dominate all the vertices of $D$. Thus, if a solution contains a vertex $v \in D$, then $v$ must dominate one or more vertices from $V(G') \setminus D$ -- otherwise we may obtain a smaller solution by deleting $v$. However, observe that the set of neighbors of $v$ outside $D$ is unaffected by deleting all the edges among the vertices in $D$. 

Note that this is a generalization of the usual \domset problem, which is obtained by setting $D = \emptyset$. We observe that Fomin et al.\ \cite{FominLS18Grid} design an \EPTAS even for \anndomset, when the input graph belongs to a family of apex-minor free graphs \footnote{Since $\cH$ is a minor-closed family, the graph obtained after deletion of edges among $D$ still belongs to $\cH$.}. In our algorithm, we use this \EPTAS as a black-box. A crucial property used to obtain an EPTAS is that the problem is contraction bidimensional. We sketch this argument, since we will subsequently use the property in our analysis.

Consider an instance $(G', D)$ of \anndomset. Suppose we contract an edge $e = uv \in E(G')$ and obtain a graph $G'' \coloneqq G' / e$. Then, let $D'' \coloneqq D \setminus \{u, v\}$, and let $(G'', D'')$ be the resulting instance of \anndomset. In other words, the vertex $w$ resulting from the contraction of the edge $uv$ is always marked as \emph{undominated}. The intuition for this is as follows. Recall that $D$ induces an independent set in $G'$. Therefore, the contracted edge $uv$ must have at least one undominated endpoint. Thus, the vertex resulting from the contraction must be marked as \emph{undominated}. We now have the following simple observations.

\begin{observation} \label{obs:contraction-closed}
	\ 
	\begin{enumerate}
		\item $\OPT(G'', D'') \le \OPT(G', D)$, where $G'' \coloneqq G' / uv$ and $D'' \coloneqq D \setminus \{u, v\}$.
		\item $D''$ forms an independent set in $G''$.
	\end{enumerate}
\end{observation}
\begin{proof}
	For item 1, let $S'$ be a solution for \anndomset instance $(G', D)$. Recall that at least one of the two vertices $u$ and $v$ must be undominated, say $u$. We consider two cases. (A) If $S' \cap \{u, v\} \neq \emptyset$, then let $S'' \coloneqq (S' \setminus \{u, v\}) \cup \{w\}$. (B) Otherwise, $S'$ contains a neighbor $u'$ of $u$. Then, we let $S'' \coloneqq S'$. Note that each of the two cases, $S''$ is a valid solution for the resulting instance $(G'', D'')$. 
	
	For item 2, we note that the vertex resulting from an edge contraction is always an undominated vertex. It follows that two dominated vertices can never become adjacent as a result of an edge contraction.
\end{proof}

We also need the following observation.
\begin{observation} \label{obs:gamma-contraction-minor}
	Let $(G', D)$ be an original instance, and suppose we obtain an instance $(\Gamma_k, D')$ from $(G', D)$ after a series of edge contractions, as defined above, where $\Gamma_k$ is the $k \times k$ triangulated grid for a sufficiently large integer $k \ge 0$ (see \Cref{sec:gammagrid} for a formal definition). Then, $\OPT(\Gamma_k, D') \ge \Omega(k^2)$.
\end{observation}
\begin{proof}
	We prove that any subgraph forming a $5 \times 5$ triangulated grid must have at least one vertex in the solution. Then, the claim follows. Consider a vertex $u \in D'$ that is at least $5$ rows and columns away from the boundary of the grid. Let $W = N(u)$ be the set of neighbors of $u$. Note that $|W| = 6$ and $W$ is a subset of the $3 \times 3$ sub-grid centered at $u$. Furthermore, \Cref{obs:contraction-closed} implies that $W \cap D' = \emptyset$, since the instance $(\Gamma, D')$ is obtained from a series of edge contractions from an original instance $(G', D)$, wherein $D$ forms an independent set. This implies that any solution must contain at least one vertex from the set of vertices $P$ forming a $5 \times 5$ sub-grid centered at $u$, since $N[W] \subseteq P$.
\end{proof}

\Cref{obs:contraction-closed} and \Cref{obs:gamma-contraction-minor} imply the following Lemma.
\begin{lemma} \label{lem:contraction-bidim}
	\anndomset is contraction bidimensional.
\end{lemma}
Additionally, we state the following result from Fomin et al.\ \cite{FominLS18Grid}.

\begin{proposition}[Corollary of Lemma 3 from \cite{FominLS18Grid}] \label{prop:sqgc}
	Let $(G', D)$ be an instance of \anndomset, where $G'$ belongs to a family of apex-minor free graphs $\cH$. Then, $\tw(G') \le \Oh\lr{\sqrt{OPT(G', D)}}$.
\end{proposition}

\medskip\noindent\textbf{\textsf{Paramaterized Approximation Scheme.}} To design an \FPTAS for \domset, we combine the case analysis from the previous section along with arguments involving bidimensionality. Recall that the input is $G = (V, E)$ and $M \subseteq V(G)$ with $|M| = p$, such that $G' \coloneqq G \setminus M \in \cH$. Let $H \coloneqq V(G) \setminus M$, and let $D \coloneqq N(M) \subseteq H$. Note that $G' = G[H]$. Let $(G', D)$ denote an instance of \anndomset, where we want to compute a dominating set for the vertices in $H \setminus D$. Since $G' \in \cH$, we use the EPTAS from \cite{FominLS18Grid} to compute a $(1+\frac{\epsilon}{2})$-approximate solution $S$. That is, $|S| \le (1+\frac{\epsilon}{2}) \OPT(G', D)$. Note that this takes $2^{\Oh(1/\epsilon)} \cdot \polyn$ time. Now, we consider two cases.

\textbf{Case 1.} $p \le \frac{\epsilon}{3} |S|$. In this case, we have the following claim.
\begin{claim}
	If $p = |M| \le \frac{\epsilon}{3} \cdot |S|$, then $M \cup S$ is a $(1+\epsilon)$-approximation for \domset on the graph $G$.
\end{claim}
\begin{proof}
	The solution $S$ dominates all vertices of $H \setminus D$. Any remaining vertex of $V(G)$ either belongs to $M$ or to $D$. However, since $D = N(M)$, the solution $M \cup S$ is a dominating set for $G$. 
	
	Now we argue about the approximation guarantee. Let $S^*$ denote an optimal \domset solution for $G$. We first argue that $S^* \setminus M$ is a feasible solution for \anndomset instance $(G', D)$. Consider any vertex $u \in H \setminus D$. If $u \not\in S^*$, then $u$ is dominated by a neighbor $v \in S^*$. Note that $v \not\in M$, since otherwise $u \in N(M) = D$, which is a contradiction. Therefore, every $u \in H \setminus D$ is dominated by some vertex of $S^* \setminus M$. Therefore, $\OPT(G', D) \le |S^* \setminus M| \le |S^*|$. Then, the claim follows, since\\$|S| + |M| \le (1+\epsilon/3) (1+\epsilon/2) \cdot \OPT(G', D) \le (1+\epsilon) \cdot|S^*|.$ \end{proof}

\textbf{Case 2.} $p > \frac{\epsilon}{3} \cdot |S| \ge \frac{\epsilon}{3} \cdot \OPT(G', D)$. In this case, we have the following claim.
\begin{claim}
	If $p = |M| > \frac{\epsilon}{3} \cdot \OPT(G', D)$, then $\tw(G) \le p + \Oh\lr{\sqrt{\frac{p}{\epsilon}}}$. 
	\\In this case, there exists a $2^{\Oh\lr{p + \sqrt{\frac{p}{\epsilon}}}} \cdot \polyn$ time algorithm to find an optimal dominating set for $G$.
\end{claim}
\begin{proof}
	Recall that \Cref{prop:sqgc} implies that $\tw(G') \le \Oh\lr{\sqrt{OPT(G', D)}}$, and due to case assumption, we have that $\OPT(G', D) \le = 3p/\epsilon$. Finally, $\tw(G) \le \tw(G \setminus M) + |M| \le p + \Oh(\sqrt{\frac{p}{\epsilon}})$. Since the treewidth of $G$ is bounded, we can use a $2^{\Oh(\tw(G))} \cdot \polyn$ time algorithm to find an optimal dominating set for $G$. 
\end{proof}

By combining the results from two cases, we get the following theorem.

\begin{theorem}
	Suppose we are given a graph $G$, and a set $M \subseteq V(G)$ of size $p$, such that $G' \coloneqq G \setminus M \in \cH$, where $\cH$ is a family of apex-minor free graphs. Then, there exists an $f(p, \epsilon)\cdot \polyn$ time algorithm to find a $(1+\epsilon)$-approximation for \domset on $G$, where $f(p, \epsilon) = \max \LR{2^{\Oh(p + \sqrt{\frac{p}{\epsilon}})}, \ \ 2^{\Oh(1/\epsilon)}}$.
\end{theorem}

%% file: maintheorem-domset.tex
%!TEX root = main.tex
\subsection{Equivalence Theorem for \domset} \label{subsec:equiv-domset}
In this section, we prove the following equivalence theorem for \domset, assuming that $\cH$ be \emph{well-behaved}, i.e., hereditary and closed under disjoint union. 
%We also require that $\cH$ satisfies the following property: let $G \in \cH$, and let $u \in V(G)$ be an arbitrary vertex. Let $G'$ be a graph obtained by attaching an arbitrary number of distinct pendant vertices to $u$. Then, $G' \in \cH$. We call a family $\cH$ satisfying these three properties an \emph{amenable} family. While the last assumption may seem arbitrary, we note that many natural graph families are amenable, e.g., (outer)planar graphs, forests, graphs of bounded treewidth, bipartite graphs, etc.

\begin{theorem} \label{thm:equiv-domset}
	Let $\cH$ be any well-behaved family of graphs, and let $\Pi= \domset$. Then, the following statements are equivalent.
	\begin{enumerate}
		\item $\Pi$ admits an \FPTAS parameterized by $\mdh(\cdot)$ and $\epsilon$.
		\item $\Pi$ admits an \FPTAS parameterized by $\edh(\cdot)$ and $\epsilon$.
		\item $\Pi$ admits an \FPTAS parameterized by $\twh(\cdot)$ and $\epsilon$.
	\end{enumerate}
\end{theorem}
\begin{proof}
	We prove $1 \implies 3$. Let $\pmd$ be the assumed \FPTAS for $\Pi$ parameterized by $\mdh$ and $\epsilon$. Let $G = (V, E)$ be the input graph on $n$ vertices, along with an $\cH$-tree decomposition $(T, \chi, L)$. Let $\ell$ denote the width of the given $\cH$-tree decomposition. First, if $\cH$ only contains the empty graph $(\emptyset, \emptyset)$, our goal is to prove that \domset admits an \FPTAS parameterized by (standard) $\tw(\cdot)$. Indeed, since \domset is \FPT parameterized by $\tw(\cdot)$, we assume that $\cH$ contains at least one non-empty graph. In particular, due to hereditary property, $\cH$ must contain a graph on a single vertex. Furthermore, we assume that $\ell \ge 1$ -- otherwise $G$ is a disjoint union of graphs belonging to $\cH$. This implies that $G \in \cH$, in which case all three statements are trivially equivalent.
	
	As before, for any leaf node $t \in V(T)$, let $V_t = \chi(t)$, $H_t = V_t \cap L$, and $R_t = V_t \setminus H_t$. Note that $G[H_t] \in \cH$, and the sets $\{H_t\}_{t \in V(T)}$ are pairwise disjoint, and $|R_t| \le \ell$ for any $t \in V(T)$. 
	
	We iterate over every leaf node $t \in V(T)$, and create a new graph $\tilG_t$ as follows. Let $v^*$ be a new vertex that is not in $V_t$. Let $V'_t = V_t \cup \{v^*\}$. Let $E'_t = E(G[V_t]) \cup \til{E}$, where $\til{E} = \LR{ v^* u : u \in R_t }$. That is, the graph $\til{G}_t$ is obtained by adding a new vertex $v^*$ to the graph $G[V_t]$ and making it adjacent to all vertices of $R_t$. 
	
	Note that $\tilG' - R_t = G[H_t] \uplus G^*$, where $G^* = ({v^*}, \emptyset)$ is an isolated component containing $v^*$. Now, $G[H_t] \in \cH$, and $G^* \in \cH$, which implies that $\tilG' - R_t = G[H_t] \uplus G^* \in \cH$, since $\cH$ is closed under disjoint union. Now, we use \pmd on $\tilG_t$ to obtain $\til{S}_t \subseteq V(\tilG_t)$, such that (i) $\til{S}_t$ is a dominating set for $\tilG_t$, and (ii) $|\tilS_t| \le (1+\epsilon/4) \cdot \OPT(\tilG_t)$. Let $S_t = \tilS_t \setminus \{v^*\}$. 
	
	For a leaf node $t \in V(T)$ if $\ell \le \frac{\epsilon}{29} \cdot |S_t|$, then we say that $t$ is a \emph{good} node; and \emph{bad} otherwise. All internal nodes are classified as bad. Let $U_g, U_b$ denote the sets of good and bad nodes respectively. Let $V_g = \bigcup_{t \in U_g} H_t$, $V_m = \bigcup_{t \in U_g} R_t$, and $V_b = V(G) \setminus (V_g \cup V_m)$. 
	
	Let $D^*$ be an optimal dominating set for $G$, and for any $t \in U_g$, let $D^*_{1, t} = H_t \cap D^*$, and $D^*_{2, t} = R_t \cap D^*$, and $D^*_t = D^*_{1, t} \cup D^*_{2, t}$. Let $D^*_1 = \bigcup_{t \in U_G} D^*_{1, t}$, and $D^*_{2} = \bigcup_{t \in U_g} D^*_{2, t}$. Note that $\LR{D^*_{1, t}}_{t \in U_g}$ is a partition of $D^*_1$, whereas a vertex in $D^*_2$ may belong to multiple $D^*_{2, t}$'s. Finally, let $D^*_b = D^* \cap V_b$. Note that $\LR{D^*_1, D^*_2, D^*_b}$ is a partition of $D^*$. 
	
	Now, let $S_1 \coloneqq \bigcup_{t \in U_g} S_t$, and $S_2 \coloneqq \bigcup_{t \in U_g} R_t$. We wish to add $S_1 \cup S_2$ to our solution. Since each $R_t$ is a separator between $H_t$ and $G - \chi(t)$, we can delete $V_g$ from the graph. However, we cannot simply delete $V_m = S_2$ from the graph -- since the vertices in $S_2$ may already dominate a subset of vertices in $V_b$. Therefore, we need to ``remember'' that we have decided to add $S_2$ into our solution. To this end, we define an \emph{annotated} version of \domset, which we call \ads (note that this is similar to, but not exactly the same as \textsc{Annotated Dominating Set} problem defined in the previous subsection). The \ads problem is defined as follows.
	
	\begin{tcolorbox}[colback=white!5!white,colframe=gray!75!black]
		\ads
		\\\textbf{Input:} An instance $(G, B, W, k)$, where $G$ is a graph, $V(G) = B \uplus W$, and $k$ is a non-negative integer
		\\\textbf{Question:} Does there exist a subset $S \subseteq V(G)$ of size at most $k$ such that (i) $S$ is a dominating set for $G$, and (ii) $B \subseteq S$?
	\end{tcolorbox}
	
	Let $F \coloneqq G[S_2 \uplus V_b]$. We want to find the smallest integer $k'$, such that $(F, S_2, V_b, k')$ is a yes-instance of \ads. Let $S_b$ denote an $(1+\epsilon/2)$-approximate solution on this instance, i.e., (i) $S_b \subseteq S_2 \cup V_b$ is a dominating set for $F$, (ii) $S_2 \subseteq S_b$, and (iii) $|S_b| \le (1+\epsilon/2) \cdot \OPT'(F)$, where $\OPT'(F)$ denotes the size of the optimal solution for the \ads instance. We prove the following technical lemma, the proof of which is rather involved, and thus given after finishing the current proof.
	\begin{restatable}{lemma}{domsetclaimC} \label{lem:domset-technical3}
		\begin{equation}
			|S| = |S_1| + |S_2| + |S_b| \le (1+\epsilon) \cdot \OPT(G)
		\end{equation}
	\end{restatable}
	
	Thus, our task is reduced to finding an approximate solution to \ads on $(F, S_2, V_b, k')$; in fact, we will find an exact solution. To do this, first we prove the following lemma, which shows that \ads can be reduced to the (standard) dominating problem on an auxiliary graph. The proof is given after finishing the current proof.
	
	\begin{restatable}{lemma}{adsTods} \label{lem:ads-to-ds}
		Let $\cI = (G, B, W, k)$ be an instance of \ads. Let $G'$ be the graph obtained by attaching $N = n^2$ distinct pendant vertices to each vertex in $B$, where $n = |V(G)|$. Let $\cI' = (G', k)$ be the resulting instance of \domset. Then $\cI$ is a yes-instance of \ads iff $\cI'$ is a yes-instance of \domset.
	\end{restatable}

	As suggested by \Cref{lem:ads-to-ds}, we create a graph $F'$, as follows. For each $u \in S_2$, let $P_u$ denote the set of $|V(F)|^2$ distinct pendant vertices attached to $u$. By slightly abusing the notation, let $(T, \chi, L)$ be the original $\cH$-tree decomposition, restricted to the vertices of $F$. We modify this to obtain an $\cH$-tree decomposition for $F'$. Consider a vertex $u \in S_2$, and note that $u \in R_t$ for some $t \in V(T)$. We pick one such arbitrary node $t$, and add a child $t_u$, which becomes a leaf in $T$. For this node $t_u$, we define its bag $\chi'(t_u) = P_u \cup \{u\}$, where $P_u$ becomes the set of \emph{base vertices} in $\chi'(t_u)$, and $u$ is a non-base vertex in $\chi'(t_u)$. Note that $G[P_u]$ is a set of isolated vertices, and thus belongs to $\cH$. For all original nodes $t \in V(T)$, we set $\chi'(t) = \chi(t)$. Finally, $L' = \bigcup_{u \in S_2} P_u$, and let $T'$ be the resulting tree. It is easy to see that $(T', \chi', L')$ is a valid $\cH$-tree decomposition for $F'$ of width at most $\ell$. 
	
	Now, we can prove an analogous version of \Cref{lem:fii-replacement} for the instance $(F, k')$ of \domset. It is known that \domset is CMSO-definable and has \FII \cite{fomin2019kernelization}. We highlight the key differences required to perform replacement using the \FII property of \domset. Note that for a node that was earlier classified as \emph{bad}, $\OPT(F[H_t]) = \OPT(G[H_t]) \le 29\ell/\epsilon$, which implies that $\OPT(F[\chi'(t)]) \le 29\ell/\epsilon + \ell \le 30\ell/\epsilon$. On the other hand, if $t$ is a newly added leaf node then $\OPT(F[\chi'(t)]) \le |R_t| \le \ell$. This implies that, when we want to find a replacement for $G[\chi'(t)]$ using an application of \Cref{lem:red2finiteindex}, we want to decide instances of the form $F[\chi'(t)] \oplus Y$, where $|V(Y)|$ is upper bounded by some function of $\ell$. Therefore, $\OPT(F[\chi'(t)] \oplus Y) \le \OPT(F[\chi'(t)]) + |V(Y)|$, which is upper bounded by some function of $\ell$ and $\epsilon$. Therefore, a version of \Cref{lem:fii-replacement} can solve \emph{the decision version} of \ads.
	
	Now, we discuss how to find a solution, using the algorithm for the decision version as an oracle. We try the following for the values of $k = |S_2|, |S_2| + 1, \ldots$, until the first time the decision version returns that $\ads(F, S_2, V(F) \setminus S_2, k)$ is a yes-instance. Let $k^*$ denote this value. 
	
	We maintain a partial solution $S_b$, which is initialized to $S_2$. Consider a leaf node $t \in V(T)$, such that $S_b$ does not already dominate all vertices of $\chi'(t) \cap L'$. We pick an arbitrary \emph{original} vertex $u \in \chi'(t) \cap L'$, and use the algorithm to decide the instance $(F, S_b \cup \LR{u}, V(F) \setminus (S_b \cup \LR{u}), k^*)$ of \ads. If the algorithm returns no, then we conclude that $u$ does not belong to a solution of size $k^*$. Otherwise, we update $S_b$ by adding $u$ to it, and proceed to the next iteration.
	
	Now we discuss the technicalities when we use the reduction from \Cref{lem:ads-to-ds} on the newly created \ads instance. In this reduction, we add $N = n^2$ distinct pendant vertices to $u$, and let $F_u$ be the resulting graph. We add all these vertices to $L'$ and to $\chi'(t)$. We also remove $u$ from $L$, i.e., move it from $H_t$ to $R_t$. Note that the resulting base graph is a disjoint union of $F'[H_t] - u$ and the $N$ pendant vertices, each of which belong to $\cH$ due to the well-behaved property. The size of the modulator increases by $1$. However, since $\OPT(F[\chi(t)])$ is bounded by $29\ell/\epsilon$, therefore, we move at most $29\ell/\epsilon$ vertices from $H_t$ to $L_t$. Thus, the width of the resulting $\cH$-tree decomposition remains bounded by $\ell+ 29\ell/\epsilon$. Furthermore, even for an incorrect choice of $u$, $\OPT(F'[\chi'(t)])$ remains bounded by a function of $\ell$ and $\epsilon$. Thus, a version of \Cref{lem:fii-replacement} can still be used to obtain a solution $S_b$ of size $k^*$ to the original instance $(F, S_2, V_b, k^*)$. This concludes the proof of the theorem.
	
	%Since $\cH$ is , the resulting base graph belongs to $\cH$, thus this is a valid $\cH$-tree decomposition. Furthermore, even for an incorrect choice of $u$, $\OPT(F'[\chi'(t)])$ still remains bounded by a function of $\ell$ and $\epsilon$. Thus, a version of \Cref{lem:fii-replacement} can still be used to obtain a solution $S_b$of size $k^*$ to the original instance $(F, S_2, V_b, k^*)$. This concludes the proof of the theorem.
\end{proof}

\subsection{Proofs of Technical Lemmas}

Now we prove the technical lemmas from the proof of \Cref{thm:equiv-domset}. For convenience, we reiterate the statement of each lemma and then provide the proof.

\domsetclaimC*
\begin{proof}
First, we prove two technical claims, then the bound follows. The first claim is stated below.
\begin{restatable}{claim}{domsetclaimA} \label{cl:domset-technical}
	The following two inequalities hold.
	\begin{equation}
		|S_2| \le \lr{1+\frac{\epsilon}{24}} \cdot (|D^*_1| + |D^*_2|) \label{eqn:domset-main0}
	\end{equation}
	\begin{equation}
		|S_1| + |S_2| \le \sum_{t \in U_g} |S_t| + |R_t| \le \lr{1+\frac{7\epsilon}{11}} \cdot (|D^*_1| + |D^*_2|) \label{eqn:domset-main1}
	\end{equation}
\end{restatable}	
\begin{proof}
	First, we consider on a particular $t \in U_g$, and focus on showing the following inequality.
	\begin{equation}
		|S_t| + |R_t| \le (1+\epsilon) \cdot |D^*_t| = (1+\epsilon) \cdot (|D^*_{1, t}| + |D^*_{2, t}|) \label{eqn:cl-domset-1}
	\end{equation}
	First, we note that $D^*_{1, t} \cup D^*_{2, t} \cup \LR{v^*}$ is a dominating set for $\tilG_t$. It follows that, 
	\begin{equation}
		\OPT(\tilG_t) \le |D^*_{1, t}| + |D^*_{2, t}| + 1\label{eqn:domset-1}
	\end{equation}
	Now, note that, 
	\begin{equation}
		\OPT(\tilG_t) \le |\tilS_t| \le \lr{1+\frac{\epsilon}{4}} \cdot \OPT(\tilG_t) \le \lr{1+\frac{\epsilon}{4}} \cdot \lr{|D^*_{1, t}| + |D^*_{2, t}| + 1} \label{eqn:domset-2}
	\end{equation}
	Now, note that since $t$ is a good node, $|S_t| \ge |\tilS_t|-1 \ge \frac{29\ell}{\epsilon}$. Therefore, $|\tilS_t| \ge \frac{28\ell}{\epsilon}$ since $\ell\ge 1$. Combining this with (\ref{eqn:domset-2}), we obtain that 
	\begin{equation*}
		2 \cdot (|D^*_{1, t}| + \ell + \ell) \ge \lr{1 + \frac{\epsilon}{4}} \cdot \lr{|D^*_{1, t}| + |D^*_{2, t}| + 1} \ge |S_t| \ge \frac{28\ell}{\epsilon}
	\end{equation*}
	Then, from the first and the last term in this sequence of inequalities, we obtain that 
	\begin{equation}
		|D^*_{1, t}| + 4\ell \ge \frac{28\ell}{\epsilon}\ \implies\ |D^*_{1, t}| \ge \frac{24\ell}{\epsilon}\ \implies\ 1 \le \ell \le \frac{\epsilon}{24} \cdot |D^*_{1,t}| \label{eqn:domset-3}
	\end{equation}
	It follows that 
	\begin{align}
		|S_t| \le |\tilS_t| &\le \lr{1 + \frac{\epsilon}{4}} \cdot (|D^*_{1,t}| + |D^*_{2,t}| +  1) \tag{From (\ref{eqn:domset-2})}
		\\&\le \lr{1 + \frac{\epsilon}{4}} \cdot \lr{ |D^*_{1,t}| + \frac{2\epsilon}{24} \cdot |D^*_{1,t}| } \tag{from (\ref{eqn:domset-3})}
		\\&\le \lr{1 + \frac{\epsilon}{2}} \cdot |D^*_{1,t}| \label{eqn:domset-4}
	\end{align}
	Finally, consider 
	\begin{align}
		|S_t| + |R_t| &\le \lr{1 + \frac{\epsilon}{11}} |S_t| \tag{Since $|R_t| \le \ell \le \frac{\epsilon}{11} |S_t|$}
		\\&\le \lr{1 + \frac{\epsilon}{2}} \cdot \lr{1 + \frac{\epsilon}{11}} \cdot |D^*_{1,t}| \tag{From (\ref{eqn:domset-4})}
		\\&\le (1+\frac{7\epsilon}{11}) \cdot |D^*_{1,t}| 
	\end{align}
	Now, adding (\ref{eqn:cl-domset-1}) over all $t \in U_g$, we obtain,
	\begin{align*}
		|S_1| + |S_2| &\le \sum_{t \in U_g} |S_t| + |R_t| 
		\\&\le \lr{1+\frac{7\epsilon}{11}} \cdot \sum_{t \in U_g} |D^*_{1, t}|
		\\&\le \lr{1+\frac{7\epsilon}{11}} \cdot |D^*_1|
	\end{align*}
	This shows (\ref{eqn:domset-main1}). 
	
	Finally, since $|R_t| \le \ell \le \frac{\epsilon}{24} \cdot |D^*_{1,t}|$, we obtain (\ref{eqn:domset-main0}) as follows.
	\begin{align*}
		|S_2| &\le \sum_{t \in U_g} |R_t| \le \sum_{t \in U_g} \frac{\epsilon}{24} \cdot \sum_{t \in U_g} |D^*_{1,t}| \le \frac{\epsilon}{24} \cdot |D^*_1| 
	\end{align*}
\end{proof}
Next, we have the following.
\begin{restatable}{claim}{domsetclaimB} \label{cl:domset-technical2}
	\begin{equation}
		|S_b| \le \lr{1+\epsilon} \cdot |D^*_b| + \frac{\epsilon}{16} \cdot |D^*_1| \label{eqn:domset-main2}
	\end{equation}
\end{restatable}
\begin{proof}
	Note that $D^*_b \cup S_2 \subseteq S_2 \cup V_b$ is a dominating set for $F$ and it contains $S_2$ as a subset. Therefore, $\OPT'(F) \le |D^*_b| + |S_2|$. Now, consider
	\begin{align*}
		|S_b| &\le \lr{1+\frac{\epsilon}{2}} \cdot \lr{|D^*_b| + |S_2|}
		\\&\le \lr{1 + \epsilon} \cdot |D^*_b| + \lr{1+\frac{\epsilon}{2}} \cdot \frac{\epsilon}{24} \cdot |D^*_1| \tag{From \Cref{cl:domset-technical}}
		\\&\le (1+\epsilon) \cdot |D^*_b| + \frac{\epsilon}{16} \cdot |D^*_1|
	\end{align*}
	This concludes the proof of the claim. 
\end{proof}
Now, we are ready to conclude the proof of \Cref{lem:domset-technical3}. From \Cref{cl:domset-technical} and \Cref{cl:domset-technical2}, we have that
\begin{align*}
	|S| = |S_1| + |S_2| + |S_b| &\le \lr{1 + \frac{7\epsilon}{11}} \cdot |D^*_1| + (1+\epsilon) \cdot |D^*_b| + \frac{\epsilon}{16} \cdot   |D^*_1|
	\\&\le (1+\epsilon) \cdot  |D^*_1|  + (1+\epsilon) \cdot |D^*_b|
	\\&\le (1+\epsilon) \cdot |D^*| \tag{Since $|D^*| = |D^*_1| + |D^*_2| + |D^*_b|$}
	\\&= (1+\epsilon) \cdot \OPT(G)
\end{align*}
\end{proof}

\adsTods*
\begin{proof}
In the forward direction, consider a dominating set $S \subseteq V(G)$ of size $k$, such that $B \subseteq S$. In the graph $G'$, $S$ dominates all vertices of $V(G)$. A vertex $u \in V(G') \setminus V(G)$ is a pendant vertex adjacent to some $v \in B$. Since $B \subseteq S$, $u$ is dominated by $v \in S$.

In the reverse direction, consider a dominating set $S' \subseteq V(G')$ of size $k \le |V(G)|$. Since $S'$ dominates all vertices of $V(G')$, in particular, it dominates the vertices of $V(G)$. Now, suppose there exists some $v \in B$ such that $v \not\in S'$. Then, all of its $n^2$ pendant neighbors must belong to $S'$, since $S'$ is a dominating set. However, since $|S'| \le k < n^2$, this is a contradiction. 
\end{proof}

%% file: connectivity.tex
%!TEX root = main.tex
\section{FPT-ASes for Connectivity Problems} \label{sec:cvc}
In this section %we give give \FPTASes for {\sc Connected Vertex Cover} and {\sc Connected Dominating Set}. Here, 
we show that bucket versus ocean also yields \FPTAS{}es for connectivity problems such as  {\sc Connected Vertex Cover} and {\sc Connected Dominating Set}.  Let $\cH$ be an apex-minor closed family. We are given a graph $G = (V, E)$, and a set $M \subseteq V(G)$ of size at most $p$ such that $G' = G\setminus M \in \cH$, Let $H = V(G) \setminus M$.

%\todo[inline]{Intro para.. in this section..}

\input{cvc-examples}
%\input{maintheorem-cvc}

%% file: cvc-examples.tex
%!TEX root = main.tex
%\subsection{FPT-ASes for Connectivity Problems} \label{subsec:cvc-fptas}

%\medskip\noindent\textbf{\textsf{Independent Set.}} By iterating over all subsets of $M$, we guess the intersection of an optimal independent set $I$ with $M$. Consider the iteration corresponding to $Y = I \cap M$. Then, let $R = N(Y) \cap H$. We compute a $(1-\epsilon)$-approximation $I'$ to the maximum independent set in the graph $G' \setminus R$, using an $2^{\Oh(1/\epsilon)} \cdot \polyn$ time algorithm from Fomin et al.~\cite{FominLS18Grid}. It follows that $|Y| + |I'| \ge |Y| + (1-\epsilon) \cdot |I \setminus Y| \ge (1-\epsilon) \cdot |I| = \OPT(G)$. Thus, we get the following result.
%
%\begin{theorem} \label{thm:indset-fptas-modh}
%	Let $\cH$ be an apex-minor free graph family. Then, there exists an \FPTAS for \indset parameterized by $\mdh(\cdot)$ that runs in time $2^{\Oh(\mdh(G) + 1/\epsilon)} \cdot \polyn$.
%\end{theorem}

%We are given a graph $G$ with a modulator $M$ to a family of graphs that excludes some fixed apex graph $H$ as minor. 
We first design \FPTAS for {\sc Connected Vertex Cover}. 
Given an $\epsilon >0$, we fix an $\epsilon'=\frac{\epsilon}{5}$. Our objective is to find a connected vertex cover $S \subseteq V(G)$ in $G$ of size $|S| \leq (1 + \epsilon) \ \OPT(G)$ where $\OPT(G)$ denotes the size of a smallest connected vertex cover of $G$. Recall that a vertex set $S \subseteq V(G)$ is called connected vertex cover if $S$ is a vertex cover of $G$ and $G[S]$ is connected. Without loss of generality, we assume that the graph is connected. Indeed,  we can safely delete isolated vertices, and if we have at least two connected components of size at least $2$, then it is trivially no-instance.   
%as we can safely remove those isolated vertices because those are not part of any optimum solution}. We will essentially &use our favorite "bucket vs ocean" procedure. 

Observe that a connected vertex cover of $G$  may not be connected vertex cover for its induced subgraphs. More elaborately, for a  graph $G$, if $G$ has a connected vertex cover $S$, then for any induced subgraph $G[Y]$ of $G$, $(S \cap Y)$ is a vertex cover but not necessarily a connected vertex cover for $G[Y]$. However, observe that every connected component $C$ of $G-M$ intersects with a neighbor of $M$. That is, 
$C\cap N(M)\neq \emptyset$. 
%On top of that, $(S \cap Y)$ has a special property that every connected component of it {\em sees} a vertex of $M$. 
This leads us to the following variant of {\sc Vertex Cover} problem which is used later in our \FPTAS.

\begin{tcolorbox}[colback=white!5!white,colframe=gray!75!black]
	\xvc
	\\\textbf{Input:} A graph $G'$, a set $X \subseteq V(G')$
	\\\textbf{Objective:} Find a minimum vertex cover, say $Z$, of $G'$ such that every connected component of $G'[Z]$ contains  a vertex from $X$. 
\end{tcolorbox}

Fomin et al.\ \cite[Section $4.1$]{FominLS18Grid} has  shown that \xvc admits an \EPTAS, with running time $2^{\Oh(1/\epsilon)} n^{\Oh(1)}$ on apex minor free graphs. 
%On the other hand, this problem is known to admit a single-exponential \FPT algorithm parameterized by the solution size via the results from Fomin et al.\ \cite{FominLMS12Deletion}. Thus, we get the following corollary.
%
%\begin{lemma}
%	\label{lem:xvc}
%	For any fixed $\epsilon'> 0$, there exist an algorithm for $\xvc$ which outputs a vertex cover $S$ of size $|S| \leq (1 + \epsilon') \OPT_{X-VC}(G')$ where $\OPT_{X-VC}(G')$ is a optimum vertex cover for $\xvc$ and we can compute it in time .......
%\end{lemma}
We use this result with $G'$ and $X=N(M)$, and obtain a solution $S'$ of  \xvc, of size $(1 + \epsilon') \OPT_{{\sf SIVC}}(G')$. Here, $\OPT_{{\sf SIVC}}(G')$ denotes the size of a smallest vertex cover, say $Z$, of $G'$ such that every connected component of $G'[Z]$ contains  a vertex from $X$. Next we prove two important inequalities. 
%useful lemmas

\begin{lemma}
	$\OPT_{{\sf SIVC}}(G') \leq \OPT(G)$
\end{lemma}

\begin{proof}
	Let $Y$ be an optimum solution for the {\sc Connected Vertex Cover} problem in $G$. Hence, $|Y| =  \OPT(G)$. Recall that $G' = G \setminus M$ and $X = N(M)$. Observe that, every connected component of $Y \cap G'$ has a neighbour in $M$.  Thus, $Y \cap G'$ is also a feasible solution for $\xvc$ of $G'$. Hence, $\OPT_{{\sf SIVC}}(G')\leq |Y \cap V(G')| \leq |Y| = \OPT(G)$. This concludes the proof of the lemma.
\end{proof}

\begin{lemma}
\label{lem:connectedvc}
%	\label{lem:makeconnected}
	$\OPT(G) \leq \OPT_{{\sf SIVC}}(G')  + 2|M|$. Further,  given a solution $S''$ of size  $\OPT_{{\sf SIVC}}(G')$, we can construct a connected vertex cover of $G$ of size at most $(|S''| +2|M|)$ in polynomial time.
\end{lemma}

\begin{proof}
Observe that $(S'' \cup M)$ is a vertex cover of $G$ but might not be a connected vertex cover of $G$. However, each connected component of $S''$ can be classified by a vertex of $M$, which implies that when we add $M$ to $S''$, then $G[S'' \cup M]$ has at most $|M|$ connected components. If $G[S''\cup M]$ is connected, we are done. Otherwise, since $G$ is connected there exists a vertex, say $v$, that has neighbors in two connected component of $G[S''\cup M]$. 
% $G'\setminus S''$ has a vertex which is neighbor to both of them. 
Add $v$ to $(S'' \cup M)$. Observe that every such addition reduces the number of connected component of $(S'' \cup M)$ by at least $1$. Since $G[S'' \cup M]$ has at most $|M|$ connected components, we need to repeat this procedure at most $(|M| -1)$ times to make it connected. At the end, the size of the connected vertex cover solution will be at most $(|S''| +2|M| - 1)$. The bound of $\OPT(G) \leq \OPT_{{\sf SIVC}}(G')  + 2|M|$ follows from the construction. 
The whole process takes polynomial time which concludes the proof.
\end{proof}

%\begin{lemma}
%	\label{lem:connectedvc}
%	$\OPT_{CVC}(G) \leq \OPT_{X-VC}(G') + 2|M|$
%\end{lemma}
%
%\begin{proof}
%	Let $Z$ be an optimum solution for the problem $\xvc$ in $G'$. Hence, $|Z| = \OPT_{X-VC}(G')$. Using Lemma~\ref{lem:makeconnected}, we construct a feasible solution of connected vertex cover problem of size at most $(\OPT_{X-VC}(G') + 2|M|)$. Hence, $\OPT_{CVC}(G) \leq \OPT_{X-VC}(G') + 2|M|$ which concludes the proof.
%\end{proof}

After getting $S'$, we compare the size of $S'$ with the size of the modulator $M$ and accordingly we divide it into two cases.

\begin{itemize}
	\item Case 1: $|M| > \epsilon'|S'|$, implies that $|S'| < |M|/ \epsilon'$.
% That means size of the approximate solution $S'$ for $\xvc$ in $G'$ is upper bounded by $f(|M|, \epsilon')$ for %some computable function $f$. 
This implies that $\OPT_{{\sf SIVC}} \leq |S'| \leq |M|/ \epsilon'$. Further, using Lemma~\ref{lem:connectedvc}, we can say that $\OPT(G) \leq \OPT_{{\sf SIVC}} + 2|M| \leq \frac{|M|}{\epsilon'}+2|M|$. Now, using the know parameterized algorithm for {\sc Connected Vertex Cover} that checks whether $G$ has a connected vertex cover of size $k$ in time 
$2^kn^{\Oh(1)}$~\cite{DBLP:conf/swat/Cygan12}, we can find a smallest connected vertex cover $S$ of $G$ in time $2^{\Oh(p/\epsilon)}n^{\Oh(1)}$. 
%
%Further using Lemma~\ref{lem:exact}, we plug in the value of $k$ to be  $g(|M|, \epsilon')$ and compute a connected vertex cover $S$ in FPT time parameterized by modulator and $\epsilon'$ . 
Hence, in this case, we solve our problem exactly.
	
\item Case 2: $|M| \leq \epsilon'|S'|$. Starting from $S'$, using Lemma~\ref{lem:connectedvc}, we construct a connected vertex cover $S$ of size at most $|S' |+2|M|\leq (1 + 2 \epsilon')|S'| \leq (1 + 2 \epsilon')(1 + \epsilon')\OPT(G) = (1 + 3 \epsilon' + 2\epsilon'^2) \OPT(G) \leq (1 + 5 \epsilon')\OPT(G) = (1 + \epsilon)\OPT(G)$. Hence, in this case, we get an $(1 + \epsilon)$ approximate solution, in time $2^{\Oh(1/\epsilon)} n^{\Oh(1)}$. 
\end{itemize}
This leads to the following theorem. 
\begin{theorem} \label{thm:cvc-fptas-modh}
	Let $\cH$ be an apex-minor free graph family. Then, there exists an \FPTAS for {\sc Connected Vertex Cover}  parameterized by $\mdh(\cdot)$ that runs in time $2^{\Oh(\mdh(G) + 1/\epsilon)} \cdot \polyn$.
\end{theorem}

The outline of \FPTAS for {\sc Connected Dominating Set} is identical to the one used for {\sc Connected Vertex Cover}. Here, we need to obtain EPTAS for following variant of  {\sc Connected Dominating Set} problem, which we call {\sc Set Intersecting Dominating Set}. Here, we are given a graph $G'$, a set of vertices $X$, and the objective is to find a minimum sized dominating set, say $Z$, that dominates all the  vertices in $V(G')\setminus X$ and every connected component of $G'[Z]$ contains a vertex of $X$. For our case $G'=G-M$ and $X=N(M)$.  Fomin et al.\ \cite[Section $4.1$]{FominLS18Grid} has  shown that this problem admits an \EPTAS, with running time $2^{\Oh(1/\epsilon)} n^{\Oh(1)}$ on  apex minor free graphs.  Finally, the result follows from the fact that $\OPT_{{\sf SIDS}}(G') \leq \OPT(G)$, $\OPT(G) \leq \OPT_{{\sf SIDS}}(G')  + 3|M|$, and that given a solution $S''$ of size  $\OPT_{{\sf SIDS}}(G')$, we can construct a connected dominating set of $G$ of size at most $(|S''| +3|M|)$ in polynomial time. This leads to the following result.

\begin{theorem} \label{thm:cvc-fptas-modh}
	Let $\cH$ be an apex-minor free graph family. Then, there exists an \FPTAS for {\sc Connected Dominating Set}  parameterized by $\mdh(\cdot)$ that runs in time $2^{\Oh(\mdh(G) + 1/\epsilon)} \cdot \polyn$.
\end{theorem}

Using ideas similar to Theorem~\ref{thm:equiv-domset}, we can get the following results. The only place we need to be careful is how to get a connected solution on graph on bad vertices. This can be achieved with careful book keeping,
%, we omit details for being boring.

\begin{theorem} \label{thm:indset-fptas}
	Let $\cH$ be a well-behaved family of graphs, and let $\Pi = ${\sc Connected Vertex Cover} 
	({\sc Connected Dominating Set}). Then, the following statements are equivalent.
	\begin{enumerate}
		\item $\Pi$ admits an \FPTAS parameterized by $\mdh(\cdot)$ and $\epsilon$.
		\item $\Pi$ admits an \FPTAS parameterized by $\edh(\cdot)$ and $\epsilon$.
		\item $\Pi$ admits an \FPTAS parameterized by $\twh(\cdot)$ and $\epsilon$.
	\end{enumerate}
\end{theorem}

%
%We start with some known results that will eventually be used in our algorithm.
%
%\begin{lemma}
%	\label{lem:exact}
%	For a given graph $G$ and a fixed integer $k > 0$, we can find a connected vertex cover of $G$ of size at most $k$ in time $2^k. n^{\mathcal{O}(1)}${\color{Red}(need to add references)}
%\end{lemma}
%%Our algorithm has two phases
%\subsection{FPT-AS for Connected Dominating Set} \label{subsec:cds-fptas}
%
%\todo[inline]{comment about connected dominating set}

%% file: conclusion.tex
\section{Conclusion}

In this paper, we have initiated a systematic exploration of the impact that recently introduced ``hybrid'' graph parameters have on the existence of good approximations for fundamental graph problems.  In fact, we have shown that as far as the task of obtaining an \FPTAS\ is concerned, for many problems,  designing an \FPTAS\ parameterized by the largest of these, i.e., $\mdh$, is sufficient  to obtain an \FPTAS\ parameterized by both $\edh$ and $\twh$ . This result gives an approximation analogue of recent equivalence obtained between these parameters in the exact algorithmic setting. 

To demonstrate concrete applicability of our techniques, we first designed {\FPTAS}es for many classical graph problems parameterized by $\mdh$, where $\cH$ is an apex/$H$-minor free graph family. Then, using our equivalence theorems, we are able to lift these {\FPTAS}es to $\edh$ and $\twh$. At this point, we would like to highlight that, in several concrete applications of our equivalence theorems, the \emph{non-uniform} \FPTASes can be made uniform. For example, we believe that one can obtain uniform \FPTASes parameterized by $\mdh, \edh$ and $\twh$ for problems such as {\sc Vertex Cover, Feedback Vertex Set, Independent Set, Dominating Set, Cycle Packing}, where $\cH$ is an apex-minor free graph family. 

We conclude with two final remarks. Firstly, the assumption of CMSO-definability can be relaxed. In fact, we only require an \EPTAS for the problem on graphs of bounded treewidth. This would then enable one to apply our framework to problems that are not CMSO-definable but are known to have good approximations on the graph family under consideration. Secondly, the initial step of our algorithm, i.e., {\FPTAS}es parameterized by $\mdh$, can be made to work with other graph families, such as (Unit) Disk Graphs, using known {\EPTAS}es for a number of graph problems \cite{FominLS18Grid,Lokshtanov22Disk}.

%% file: fii-theorem-proof.tex
%!TEX root = main.tex

\section{Background on \FII and Proof of \Cref{lem:fii-replacement}} \label{sec:fii-background}
The material from \Cref{subsec:boungrap} to \Cref{subsec:replacement} is largely borrowed verbatim from \cite{AgrawalKLPRSZ22Elimination}.
\input{background}

\subsection{Proof of \Cref{lem:fii-replacement}} \label{subsec:proof}
For convenience, we restate the lemma.
\LemmaFII*
\begin{proof}
	Note that since $\Pi$ is self-reducible (cf. \Cref{lem:self-reducibility}), given an algorithm for the decision version, we can find an optimal solution at the expense of polynomial overhead in the running time. Thus, we focus on the decision version. We perform an iterative search on $k$, by trying values $0, 1, \ldots$, until the following algorithm concludes that $(G, k)$ is a yes-instance for the first time.
	
	Now, fix a positive integer $\ell$ as in the statement of the lemma. We will assume that the constants $\xi_i$ from \Cref{lem:red2finiteindex} for each $i \in [\ell+1]$ are hardcoded in the algorithm. We want to gradually transform the graph $G$ to another graph $G'$, such that the number of vertices of $L$ in each leaf bag in the $\cH$-tree decomposition of $G'$ is upper bounded by $\mu \coloneqq \sum_{i \in [\ell+1]} \xi_i$. Then, the (standard) treewidth of $G'$ is upper bounded by $\ell+ \mu$. We will achieve this by a sequence of replacements via \Cref{lem:red2finiteindex} for each $\xi(t)$, using the fact that $\Pi$ has \FII. Now, we proceed to the formal proof.
	
	Let $\til{A}$ be the set of nodes in $T$ whose bags contain more than $\mu$ vertices from $L$, i.e., $\til{A} \coloneqq \LR{ v \in V(T) : |\chi(t) \cap L| > \mu }$. Furthermore, let $\til{\cG} \coloneqq \LR{ G[\chi(t) \cap L] : t \in \til{A} }$, i.e., $\til{\cG}$ is the set of graphs induced by vertices in $L$ in each of the bags of nodes in $\til{A}$. Let $\til{\cG} = \LR{ \tilG_1, \tilG_2, \ldots, \tilG_q }$, where the graphs are numbered arbitrarily.
	
	We create a sequence of graphs $G_0, G_1, \ldots, G_q$, and a sequence of constants $c_0, c_1, \ldots, c_q$ as follows. Let $G_0 = G$ and $c_0 = 0$. Now, we iterate over $i \in [q]$, and proceed as follows. For graph $\tilG_i$, let $\tilt_i$ be the unique leaf in $T$ such that $V(\tilG_i) \subseteq \chi(\tilt_i)$, and let $\tilb_i \coloneqq \chi(\tilt_i) \setminus V(\tilG_i)$, and $b_i \coloneqq |\tilb_i|$. Note that due to properties of $\cH$-tree decomposition, $b_i \le \ell+1$. Let $\lambda_{\tilG_i}: \tilb \to [b_i]$ be an arbitrary injective function, and by slightly abusing the notation, let $\tilG_i$ denote the resulting boundaried graph. Note that $V(\tilG_i) \subseteq V(G_{i-1})$. Similarly, let $G'_i$ be the boundaried graph $G_{i-1} - V(G_i)$ with boundary $\tilb_i$ with the same mapping $\lambda_{\tilG_i}$. Note that $\tilG_{i-1} = \tilG_i \oplus G'_i$. 
	
	Now, we want to apply \Cref{lem:red2finiteindex} w.r.t.\ $\tilG_i$ to find a replacement graph $\tilG^*_i$, and a translation constant $c_i$, such that $\tilG_i \equiv_\Pi \tilG^*_i$, and $|V(\tilG^*_i)| \le \xi_{b_i}$. We observe that, in the proof of \Cref{lem:red2finiteindex}, in order to find such a pair $(\tilG^*_i, b_i)$, we need to decide instances where the value of an optimal solution is bounded by $\OPT(\tilG_i \oplus Y_j) \le \OPT(\tilG_i) + \max_{j} |Y_j| \le 3\ell/\epsilon + \ell + \mu$, which follows from the assumption $\star$ from the statement of the lemma.
	
	Note that $\pmd$ decides any instance $G'$ of $\Pi$ exactly, as long as $\OPT_\Pi(G') \le 3\ell/\epsilon + \mu$, if the approximation parameter $\epsilon'$ is set to be at most $\le \frac{1}{2 \cdot (3\ell/\epsilon + \mu)}$. Let $\cP'$ denote the resulting algorithm where $\epsilon'$ is set to be $\frac{1}{2 \cdot (3\ell/\epsilon + \mu)}$. Then, the running time of the algorithm is upper bounded by $f(\ell, \epsilon') \cdot \polyn = f'(\ell, \epsilon) \cdot \polyn$, for some computable function $f'$. As a result, each application of \Cref{lem:red2finiteindex} takes time $f''(\ell, \epsilon) \cdot \polyn$ for some computable function $f''$. 
	
	Returning to the proof, let $\tilG^*_i$ be the graph returned by \Cref{lem:red2finiteindex} such that $\tilG_i \equiv_\Pi \tilG^*_i$, and $c_i$ be the translation constant. Let $G_i = \tilG^*_i \oplus G'_i$. Let $c^* = \sum_{i \in [q]} c_i$. We now prove the following statements.
	\begin{enumerate}
		\item The instance $(G_q, k+c^*)$ can be constructed in time $f''(\ell, \epsilon) \cdot \polyn$.
		\item $(G_q, k+c^*)$ and $(G, k)$ are equivalent instances of $\Pi$.
		\item $\tw(G_q) \le \ell +  \mu$. 
	\end{enumerate}
	Note that the set $\til{\cG}$ can be computed in polynomial time, and $q \le n$. Further, each application of \Cref{lem:red2finiteindex} takes time $g(\ell, \epsilon) \cdot \polyn$, as argued previously. Therefore, the graph $(G_q, k+c^*)$ can be computed in time $f''(\ell, \epsilon) \cdot \polyn$.
	
	Now we argue inductively that for each $i \in [q] \cup \{0\}$, $(G_i, k + \sum_{j \in [i] \cup \{0\}} c_j)$ and $(G, k)$ are equivalent instances of $\Pi$. When $i = 0$, $G_0 = G$, and $k + c_0 = k$, so the claim trivially follows. Next, assume that the claim is true for some $1 \le i-1 \le q-1$, i.e., $(G_{i-1}, k + \sum_{j \in [i-1]} c_j)$ and $(G, k)$ are equivalent instances of $\Pi$. We now prove the statement for $i$. By construction, $G[V(\tilG_{i})] = G_{i-1}[V(\tilG_i)] = \tilG_i$, and $G'_i = G_{i-1} - V(\tilG_i)$ are boundaried graphs with boundary $\tilb_i$, and $G_i = \tilG_i \oplus G'_i$. Furthermore, \Cref{lem:red2finiteindex} guarantees that $\tilG^*_i \equiv_{\Pi} \tilG_i$. Therefore, $(G_{i-1}, k + \sum_{j \in [i-1]} c_j)$ and $(G, k + \sum_{j \in [i]}c_j)$ are equivalent instances of $\Pi$. 
	
	Now, we prove the third property. To this end, consider the $\cH$-tree decomposition $(T, \chi, L)$ of $G$. For each leaf node $t_i \in V(T)$ corresponding to a graph $\tilG_i \in \til{\cG}$, we set $\chi'(t) = \chi(t) \cup (V(\tilG^*_i) \setminus L)$. For all other nodes $t \in V(T)$, let $\chi'(t) = \xi(t)$. Let $(T, \chi')$ be the resulting (standard) tree decomposition of $G_q$. The bound on the treewidth follows since \Cref{lem:red2finiteindex} guarantees $|V(\tilG^*_i)| \le \xi_{b_i} \le \mu$. 
	
	Note that $\mu$ is a constant that depends only on $\ell$, thus $\tw(G_q)$ is bounded by a function of $\ell$. Finally, since $\Pi$ is CMSO-definable, we can decide the instance $(G_q, k+c^*)$ of $\Pi$ using Courcelle's theorem \cite{Courcelle90} in time $h(\ell) \cdot \polyn$. Thus, the theorem follows.
\end{proof}

\subsection{Grids and triangulated grids} \label{sec:gammagrid}
Given a $k\in\Bbb{N}$, we denote by $\boxplus_{k}$ the $(k\times k)$-grid that is  the graph with vertex set $\{(x,y) \mid  x,y \in\{1,\dots, t\}\}$ and where two different vertices $(x,y)$ and $(x',y')$ are adjacent if and only if $|x-x'|+|y-y'| = 1$. Notice that $\boxplus_k$ has exactly $k^2$ vertices.

For  $k\in\Bbb{N}$, the graph $\Gamma_k$ is obtained from the grid $\boxplus_k$ by adding, for all $1 \leq x,y \leq k-1$, the edge with endpoints $(x+1,y)$ and $(x,y+1)$ and additionally making vertex $(k,k)$ adjacent to all the other vertices $(x,y)$ with $x \in \{1,k\}$ or $y \in \{1,k\}$, i.e., to the whole perimetric border  of $\boxplus_k$. Graph $\Gamma_9$ is shown in Fig.~\ref{fig-gamma-reg}. 
%The graph $\Gamma_{k}$ has been defined in~\cite{FominGT11cont} in the context of bidimensionality theory.

\begin{figure}
 \begin{center}
\scalebox{.6}{\includegraphics{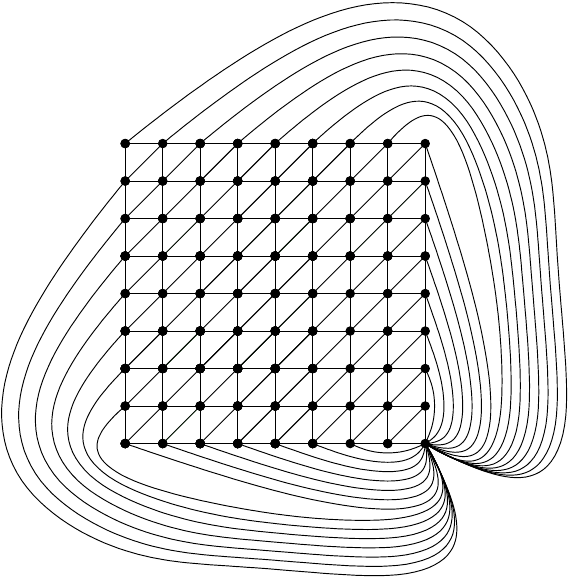}}
\caption{Graph $\Gamma_{9}$.}\label{fig-gamma-reg}
 \end{center}
 \end{figure}

%% file: background.tex
%!TEX root = main.tex
\renewcommand{\Bbb}[1]{\mathbb{#1}}

\subsection{Boundaried Graphs} 
\label{subsec:boungrap}
Here we define the notion of {\em boundaried graphs} and various operations on them.
\begin{definition}{\rm [\bf Boundaried Graphs]}\label{def:boungraph}
	A boundaried graph is a graph $G$ with a set $B\subseteq V(G)$ 
	of  distinguished vertices and an injective labelling $\lambda_G$ 
	from $B$  to the set $\Bbb{Z}^{+}$. The set $B$ is called the {\em boundary} of $G$ and  the vertices in $B$  are called  {\em boundary vertices} or {\em terminals}. 
	Given a boundaried graph $G,$ we denote its boundary by ${\delta(G)},$
	we denote its labelling by $\lambda_G$, 
	and we define its {\em label set} by $\Lambda(G)=\{\lambda_{G}(v)\mid v\in \delta(G)\}$.
	Given a finite set $I\subseteq \Bbb{Z}^{+}$, we define 
	${{\cal F}_{I}}$  to denote the class of all boundaried graphs whose label set is $I$. 
	%Similarly, we define ${\cal F}_{\subseteq I}=\bigcup_{I'\subseteq I}{\cal F}_{I'}$.
	We also denote by ${{\cal F}}$ the class of all boundaried graphs.
	Finally we say that a boundaried graph is a {\em $t$-boundaried} graph if $\Lambda(G)\subseteq \{1,\ldots,t\}$.
\end{definition}
%\todo[inline]{fix the mark in the definition}
%
%We remark that  in the labelling of the boundary of a $t$-boundaried graph, not all $t$ available labels are necessary used.

%
%For a graph $G=(V,E)$ and a vertex set $S \subseteq V,$ we sometime consider the graph $G[S]$ as the 
%$|\partial(S)|$-boundaried graph with $\partial(S)$ being the boundary.

\begin{definition}{\rm [\bf Gluing by $\oplus$]} Let $G_1$ and $G_2$ be two  boundaried graphs. We denote by $G_1 {\oplus} G_2$ the  graph 
	(not boundaried) obtained by taking the disjoint union of $G_1$ and $G_2$ and identifying equally-labeled vertices of the boundaries of $G_{1}$ and $G_{2}.$ In $G_1 \oplus G_2$ there is an edge between two vertices if there is  an edge between them either in $G_1$ or in $G_2$, or both.  
	
\end{definition}

We remark that if $G_1$ has a label which is not present in $G_2$, or vice-versa, then in $G_1 \oplus G_2$ we just forget that label. 

\begin{definition} {\rm [\bf Gluing by $\oplus_\delta$]}
	The {\em boundaried gluing operation} $\oplus_{\delta}$ is similar to the normal gluing operation, but results in a boundaried graph rather than a graph. Specifically $G_1 \oplus_\delta G_2$ results in a boundaried graph where the graph is $G = G_1 \oplus G_2$ and a vertex is in the boundary of $G$ if it was in the boundary of $G_1$ or  of $G_2$. Vertices in the boundary of $G$ keep their label from $G_1$ or $G_2$. 
	%Both for gluing and boundaried gluing we will refer to $G_1 \oplus G_2$ or $G_1 \oplus_\delta G_2$ as the {\em sum} of $G_1$ and $G_2$, and %$G_1$ and $G_2$ are the {\em terms} of the sum.
	%\todo[inline]{we do not need this}
\end{definition}

%\begin{definition}
%Let $G=G_{1}\oplus G_{2}$ where $G_{1}$ and $G_{2}$ are boundaried graphs.
%We define the {{\em glued}} set of $G_{i}$ as the set $B_{i}=\lambda_{G_{i}}^{-1}(\Lambda(G_{1})\cap \Lambda(G_{2})), i=1,2$. For a vertex $v\in V(G_{1})$ we define its {{\em heir}} $\mar{h(v)}$ in 
%$G$ as follows: if $v\not\in B_{1}$ then $h(v)=v$, otherwise $h(v)$ is the result of the identification 
%of $v$ with an equally labeled vertex in $G_{2}$. The {\em heir} of a vertex in $G_{2}$ is defined symmetrically. The {{\em common boundary}} of $G_{1}$ and $G_{2}$ in $G$ is equal 
%to $h(B_{1})=h(B_{2})$ where the evaluation of $h$ on vertex sets is defined in the obvious way.
%The {\em heir} of an edge $\{u,v\}\in E(G_{i})$ is the edge $\{h(u),h(v)\}$ in $G$.
%\end{definition}

Let ${\cal G}$ be a class of (not boundaried)  graphs.
By slightly abusing notation we say that a boundaried graph {\em belongs to a graph class ${\cal G}$} if the underlying graph belongs to ${\cal G}.$

\begin{definition}{\rm [\bf Replacement]}\label{defn:replacement}
	Let $G$ be a $t$-boundaried graph containing a set $X\subseteq V(G)$
	such that $\partial_{G}(X)=\delta(G).$ Let $G_1$ be a $t$-boundaried graph. The result of {\em replacing $X$ with $G_1$} is the graph $G^{\star}\oplus G_{1},$
	where $G^{\star}=G\setminus (X\setminus \partial (X))$ is treated as a $t$-boundaried graph with  $\delta(G^{\star})=\delta(G).$
\end{definition}

	\subsection{Finite Integer Index}
	\label{subsec:finiinteginde}
	\begin{definition}{\rm [\bf Canonical equivalence on boundaried graphs.]}
		Let $\Pi$ be a parameterized graph problem whose instances are pairs of the form $(G,k).$
		Given two boundaried graphs $G_1,G_2~\in {\cal F},$ we say that $G_1\!\equiv _{\Pi}\! G_2$ if 
		$\Lambda(G_{1})=\Lambda(G_{2})$
		and there exists a {\em transposition constant}
		$c\in\Bbb{Z}$ such that 
		\begin{eqnarray*}
			\forall(F,k)\in {\cal F}\times \Bbb{Z} &&  (G_1 \oplus F, k) \in \Pi \Leftrightarrow (G_2 \oplus F, k+c) \in \Pi.\label{eq:fiidef}
		\end{eqnarray*}
		Here, $c$ is a function of the two graphs $G_1$ and $G_2$. 
		%\sed{$c$ can be a positive or negative integer?}
		%\end{itemize}
	\end{definition}
	Note that  the relation $\equiv_{\Pi}$  is
	an equivalence relation. Observe that $c$ could be negative in the above definition. This is the reason we allow the parameter in parameterized problem instances to take negative values.

	Next  we define a notion of ``transposition-minimality'' for the members 
	of  each equivalence class of $\equiv_{\Pi}.$

	\begin{definition}{\rm [\bf Progressive representatives~\cite{BodlaenderFLPST16}]}
		\label{def:progrepr}
		Let $\Pi$ be a parameterized graph problem whose instances are pairs of the form $(G,k)$
		and let ${\cal C}$ be some equivalence class of $\equiv_{\Pi}$. We say that $J\in{\cal C}$ is a {{\em progressive 
				representative}}
		of ${\cal C}$ if for every $H\in{\cal C}$
		there exists $c\in\Bbb{Z}^{-},$ such that 
		\begin{eqnarray}
			\forall(F,k)\in {\cal F}\times \Bbb{Z} \ \ \  (H \oplus F, k) \in \Pi \Leftrightarrow (J\oplus F, k+c) \in \Pi. \label{eq:progfii}
		\end{eqnarray}
	\end{definition}
	
	The following lemma guarantees the existence of a progressive representative for each equivalence class of 
	$\equiv_{\Pi}$. 
	%Consider a graph $H$ in the equivalence class that 

	\begin{lemma}[\cite{BodlaenderFLPST16}]
		\label{lem:existprog}
		Let $\Pi$ be a parameterized graph problem whose instances are pairs of the form $(G,k)$.
		%and let $t\in\Bbb{Z}^{+}.$  
		Then each  equivalence class of $\equiv_{\Pi}$ has a progressive representative.
	\end{lemma}

	Notice that two  boundaried graphs with different label sets belong to 
	different equivalence classes of $\equiv_{\Pi}.$ Hence for every equivalence 
	class ${\cal C}$ of $\equiv_{\Pi}$ there exists some finite set $I\subseteq\Bbb{Z}^{+}$ such that 
	${\cal C}\subseteq  {\cal F}_{I}$. We are now in position  to give the following definition.
	
	\begin{definition}{\rm [\bf Finite Integer Index]}
		\label{def:deffii}
		A parameterized graph problem $\Pi$ whose instances are pairs of the form $(G,k)$
		has {\em Finite Integer Index} (or  is {{\em FII}}), if and only if for every finite $I\subseteq \Bbb{Z}^+,$
		the number of equivalence classes of  $\equiv_{\Pi}$ that are subsets of ${\cal F}_{I}$
		is finite. For each $I\subseteq \Bbb{Z}^{+},$ we define ${\cal S}_I$ to be
		a set containing exactly one progressive representative of each equivalence class of $\equiv_{\Pi}$
		that is a subset of ${\cal F}_{ I}$. We also define ${\cal S}_{\subseteq I}=\bigcup_{I'\subseteq I} {\cal S}_{I'}$. 
	\end{definition}

		\subsection{Replacement lemma} \label{subsec:replacement}
		%The result of this section will be applicable in replacing  the following kind of protrusions. 
		%\begin{Definition}{\rm [\bf $r$-$\Pi$-protrusion]} Let $\Pi$ be a {\sc $p$-min-CMSO} vertex subset problem. 
		% Given a graph $G$, we say that a set $X\subseteq V(G)$ is an {\em $r$-$\Pi$-protrusion} of $G$ if 
		%   the number of vertices in $X$ with a neighbor in $V(G)\setminus X$ is at most $r$ and there exists a 
		%   subset $S\subseteq X$ of size at most $r$ such that  $(G[X],S)\models \psi.$. 
		%\end{Definition}
		%
		This subsection is verbatim taken from Fomin et al.~\cite[Section $3.3$]{FominLST18} and is provided here only for completion. We only need to make few simple modifications to suit our need.

		\begin{definition}
			Let $\cal G$ denote the set of all graphs. A graph {\em parameter}  is a function  $\Psi \colon \cal G \to \mathbb{Z}^{+}$.  That is, $\Psi$ associates a non-negative integer to a graph $G \in \cal G$. The parameter $\psi$ is called {\em monotone}, if for every $G \in \cal G$, and for every $V_1\subseteq V_2$, 
			$\Psi(G[V_2])\geq \Psi(G[V_1])$.  
		\end{definition}
		
		We can use $\Psi$ to define several graph parameters such as {\em treewidth}, or given a family $\cal F$ of graphs, a minimum sized vertex subset $S$ of $G$, called modulator, such that $G-S\in {\cal F}$.  
		Next we define a notion of monotonicity for parameterized problems. 
		
		\begin{definition}{\rm (\cite[Definition $3.9$]{FominLST18}).}
			We say that a parameterized graph problem $\Pi$ is {\em positive monotone} if for every graph $G$ 
			there exists a unique $\ell \in \Bbb{N}$ such that for all $\ell'\in \mathbb{N}$ and $\ell' \geq \ell$, $(G,\ell')\in \Pi$ and for all 
			$\ell'\in \mathbb{N}$ and $\ell' < \ell$, $(G,\ell')\notin \Pi$.  A parameterized graph problem $\Pi$ is {\em negative monotone} if for every graph $G$ 
			there exists a unique $\ell \in \Bbb{N}$ such that for all $\ell'\in \mathbb{N}$ and $\ell' \geq \ell$, $(G,\ell')\notin \Pi$ and for all 
			$\ell'\in \mathbb{N}$ and $\ell' < \ell$, $(G,\ell')\in \Pi$. $\Pi$ is monotone if it is either positive monotone or negative monotone.  
			We denote the integer $\ell$ by {\sc Threshold($G,\Pi$)} (in short  {\sc Thr($G,\Pi$)}). 
		\end{definition}

		We first give an intuition for the next definition.  We are considering monotone functions and thus for every graph $G$ 
		there is an integer $k$ where the answer flips. However, for our purpose we need a corresponding notion for 
		boundaried graphs.   If we think of the representatives as some ``small perturbation'', then it is the max threshold over all small perturbations (``adding a representative = small perturbation''). This leads to the following definition. 
		
		\begin{definition}{\rm (\cite[Definition $3.10$]{FominLST18}).}\label{def:equiv-boundaried-graph}
			Let $\Pi$ be a monotone parameterized graph problem that has {\FII} and $\Psi$ be a graph parameter. Let  ${\cal S}_t$  be
			a set containing exactly one progressive representative of each equivalence class of $\equiv_{\Pi}$ that is a subset of 
			${\cal F}_{I}$, where $I=\{1,\ldots,t\}$.  
			For a $t$-boundaried graph $G$, we define   
			\begin{eqnarray*}
				\iota(G) & = &  \max_{G'\in {\cal S}_t}  \mbox{{\sc Thr($G\oplus G',\Pi$)}},\\
				\mu(G) & = &  \max_{G'\in {\cal S}_t}  \Psi(G\oplus G') 
			\end{eqnarray*}
		\end{definition}
		
		%\todo[inline]{do we change the kappa definition here --- as this conflicts with the similar definition in treewidth section}
		
		The next lemma says the following. Suppose we are dealing with some {\FII} problem and we are given a boundaried graph with boundary size $t$.  We know it has a representative of size $h(t)$ and we want to find this representative. In general finding a representative for a boundaried graph is more difficult than solving the corresponding problem. 
		The next lemma says basically that if  we can find the  ``OPT'' of a boundaried graph efficiently 
		then we can efficiently find its representative. Here by ``OPT''  we mean $\iota(G)$, which is a robust version of the threshold function (under adding a representative). 
		And by efficiently we mean as efficiently as solving the problem on normal (unboundaried) graphs.

		\begin{lemma}{\rm (\cite[Lemma $3.11$]{FominLST18}).}
			%\todo{The running time should be independent of $k$ and updated accotdingly}
			\label{lem:red2finiteindex}
			Let $\Pi$ be a monotone parameterized graph problem that has {\FII} and $\Psi$ be a graph parameter.  Furthermore, let $\cal A$ be an algorithm for $\Pi$ that, given a pair $(G,k)$, decides whether it is in $\Pi$ in time $f(|V(G)|,\Psi(G))$. 
			Then for every $t\in\Bbb{N},$ there exists a $ \xi_t \in\Bbb{Z}^{+}$ (depending on $\Pi$ and $t$), and 
			an algorithm that, given a $t$-boundaried graph $G$  with $|V(G)|>\xi_t,$ outputs, in  $\Oh(\iota(G)(f(|V(G)|+\xi_t,\mu(G)))$ steps,
			a $t$-boundaried graph $G^\star$  such that $G\equiv_{\Pi}G^\star$ and  $|V(G^\star)| < \xi_t$. Moreover we can compute the translation 
			constant  $c$ from $G$ to $G^\star$ in the same time.
		\end{lemma}
		
		\begin{proof}
			We give prove the claim for positive monotone problems $\Pi$; the proof for negative monotone problems is identical. 
			%Let $G^*$ be a progressive representative of the equivalence class $C$ of $\equiv_{\Pi}$ to which $G$ belongs. Since $%\equiv_{\Pi}$ has finite index we get that there is a finite number 
			%Recall that we denote by  ${\cal S}_{t}$  a set of (progressive) representatives for $(\Pi, t)$ and let
			% $\xi_t=\max_{Y\in {\cal S}_{t}}|Y|.$ 
			Let  ${\cal S}_t$  be
			a set containing exactly one progressive representative of each equivalence class of $\equiv_{\Pi}$ that is a subset of 
			${\cal F}_{I}$, where $I=\{1,\ldots,t\}$, and  let  $\xi_t=\max_{Y\in {\cal S}_{t}}|V(Y)|.$ The set  ${\cal S}_{t}$ is hardwired in the description of the algorithm. 
			Let $Y_1,\ldots,Y_\rho$ be the set of progressive representatives in ${\cal S}_{t}$. Let ${\cal F}_{t}={\cal F}_{I}$. Our objective is to find 
			a representative $Y_\ell$  for  $G$ such that 
			\begin{eqnarray}
				\forall (F,k)\in {\cal F}_{t}
				\times \Bbb{Z} & &   (G \oplus F, k) \in \Pi  \Leftrightarrow    (Y_\ell \oplus F, k-\vartheta(X,Y_\ell)) \in \Pi. 
				\label{eq:progresivereplacement}
			\end{eqnarray}
			Here, $\vartheta(X,Y_\ell)$ is a constant  that depends on $G$ and $Y_\ell$.  Towards this   
			we define the following matrix for the set of representatives. Let 
			$$A[Y_i, Y_j]=  \mbox{{\sc Thr($Y_i\oplus Y_j,\Pi$)}}$$
			The size of the matrix $A$ only depends on $\Pi$ and $t$ and is also hardwired in the description of the algorithm.  Now given $G$ we find its representative as follows. 
			\begin{itemize}
				\item Compute the following row vector ${\cal X}=[ \mbox{{\sc Thr($G\oplus Y_1,\Pi$)}}, \ldots ,  
				\mbox{{\sc Thr($G\oplus Y_\rho,\Pi$)}})]$. For each $Y_i$ we decide whether $(G\oplus Y_i,k)\in \Pi$ using the assumed algorithm for deciding 
				$\Pi$,  letting $k$ increase from $1$ until the first time $(G\oplus Y_i,k)\in \Pi$. Since $\Pi$ is positive monotone this will happen for some 
				$k\leq \iota(G)$. Thus the total time to compute the vector ${\cal X}$ is $\Oh(\iota(G)(f(|V(G)|+\xi_t,\mu(G)))$. 
				
				\item Find a translate row in the matrix $A(\Pi)$. That is, find an integer $n_o$ and a representative 
				$Y_\ell$ such that  
				\begin{eqnarray*}
					[ \mbox{{\sc Thr($G\oplus Y_1,\Pi$)}},  \mbox{{\sc Thr($G\oplus Y_2,\Pi$)}}, \ldots ,  
					\mbox{{\sc Thr($G\oplus Y_\rho,\Pi$)}}] \\
					=[ \mbox{{\sc Thr($Y_\ell\oplus Y_1,\Pi$)}}+n_0,  \mbox{{\sc Thr($Y_\ell\oplus Y_2,\Pi$)}}+n_0, \ldots ,  
					\mbox{{\sc Thr($Y_\ell\oplus Y_\rho,\Pi$)}}+n_0]
				\end{eqnarray*}
				Such a row must exist since ${\cal S}_t$ is  a set of representatives for $\Pi$; the representative $Y_\ell$ for the equivalence class to which $G$ belongs, satisfies the condition.  
				\item Set $Y_\ell$ to be $G^\star$ and the translation constant to be $-n_0$.
			\end{itemize}
			From here it easily follows that $G\equiv_{\Pi}G^\star$. This completes the proof.  
		\end{proof}
		We remark that the algorithm whose existence is guaranteed by the Lemma~\ref{lem:red2finiteindex} assumes that the set  ${\cal S}_{t}$ of representatives  are hardwired in the algorithm.  In its full generality we currently do not know of a procedure that for problems having {\FII} outputs such a representative set. Thus, the algorithms using Lemma~\ref{lem:red2finiteindex}  are not uniform. 
		
		\iffalse
		Next we illustrate a situation in which one can  can apply  Lemma~\ref{lem:red2finiteindex} to reduce a portion of a graph. Let $\cal F$ be a family of interval graphs. Further, let $\Pi$ be the {\sc Dominating Set} problem and $\Psi$ denote the modulator to $\cal F$. That is, given a graph $G$, 
		$$\Psi(G)=\min_{ S\subseteq V(G), G-S \in {\cal F}} |S|. $$
		It is possible to show that {\sc Dominating Set} parameterized by $\Psi(G)$ is \FPT. That is, we can design an algorithm that can decide whether an instance $(G,k)$ of   {\sc Dominating Set} is an yes-instance in time 
		$f(\Psi(G))\cdot n^{\Oh(1)}$. In fact, in time  $2^{\Oh(\Psi(G))}n^{\Oh(1)}$. This implies that if we have a 
		$t$-boundaried graph $G$, then we can find a representative of it with respect to {\sc Dominating Set} in time   $2^{\Oh(\mu(G))}n^{\Oh(1)}$. We will see its uses in this way in Section~\ref{sec:crossParam}. 
		\fi